\documentclass[12pt]{article}

\usepackage[margin=1in]{geometry}
\usepackage{amsmath,amssymb,amsthm}
\usepackage{tikz}
\usetikzlibrary{arrows}
\usetikzlibrary{decorations}
\usetikzlibrary{shapes.misc}
\usepackage{float}
\usepackage{appendix}
\usepackage{url}
\usepackage{refcount}
\usepackage{xspace}
\usepackage{subfigure}
\usepackage{nicefrac}
\usepackage[inline]{enumitem}
\usepackage[mathscr]{euscript}

\usepackage{latexsym}
\usepackage{color}
\usepackage{graphicx}
\usepackage{multirow}

\newcommand{\old}[1]{}

\renewcommand{\emph}[1]{\textit{#1}}
\newsavebox{\mycases}

 \allowdisplaybreaks

\def\black{\color{black}}

\newcounter{rot}

\newcommand{\ignore}[1]{}

\def\ii_(#1,#2){i_{#1}^{#2}}

\def\a{\alpha}
\def\b{\beta}
\def\d{\delta}

\def\f{\phi}
\def\F{\Phi}

\def\Th{\Theta}
\def\l{\lambda}

\def\p{\pi}

\def\r{\rho}

\def\s{\sigma}

\def\Om{\Omega}

\def\x{\xi}

\newcommand{\ol}[1]{\overline{#1}}

\newcommand{\ooi}{(1+o(1))}

\newcommand{\brac}[1]{\left( #1 \right)}

\renewcommand{\Pr}{{\bf Pr}}
\newcommand\bfrac[2]{\left(\frac{#1}{#2}\right)}

\newcommand{\nospace}[1]{}

\newcommand{\beq}[1]{\begin{equation}\label{#1}}
\def\eeq{\end{equation}}

\def\es{\emptyset}

\setlist[itemize,1]{leftmargin=2cm,labelsep=1cm,itemsep=20pt,topsep=10pt}
\setlist[enumerate,1]{topsep=5pt,itemsep=0pt,label=(\alph*)}

\floatstyle{plain}
\restylefloat{figure}

\newtheorem{theorem}{Theorem}
\newtheorem{lemma}[theorem]{Lemma}

\theoremstyle{remark}
\newtheorem{remark}{Remark}

\newcommand{\onept}{\hspace*{1pt}}

\newcommand{\E}{\mathbf{E}}

\let\epsilon=\varepsilon

\renewcommand{\thesubfigure}{(\roman{subfigure})}

\tikzset{every picture/.style={line width=0.8pt}}
\tikzset{empty/.style={rectangle,draw=none,fill=none}}
\tikzset{bnode/.style={circle,draw,inner sep=0pt,fill=white,minimum size=1.5mm}}
\tikzset{rnode/.style={circle,draw,inner sep=0pt,fill=black,minimum size=1.5mm}}

\begin{document}
\makeatletter
\title{
Discordant voting processes on finite graphs.\thanks{
This work was supported by EPSRC grant EP/M005038/1,
``Randomized algorithms for computer networks'', NSF grant DMS0753472  and Becas CHILE.
}
   }
\author{
Colin Cooper\thanks{Department of Informatics, King's College London, UK.
{\tt colin.cooper@kcl.ac.uk}}
\and Martin Dyer\thanks{School of Computing, University of Leeds, Leeds, UK.
{\tt M.E.Dyer@leeds.ac.uk }}
\and Alan
Frieze\thanks{Department of Mathematical Sciences, Carnegie Mellon
University, Pittsburgh PA 15213, USA. {\tt alan@random.math.cmu.edu }}
\and Nicol\'as Rivera\thanks{Department of Informatics, King's College London, UK.
{\tt nicolas.rivera@kcl.ac.uk}}
}
\maketitle \makeatother

\begin{abstract}
We consider an asynchronous  voting process on graphs  called discordant voting, which can be described as follows. Initially each vertex holds one of two opinions, red or blue.  Neighbouring vertices with different opinions
interact pairwise along an edge.
After an interaction both vertices have the same colour. The quantity of interest is the time to reach consensus, i.e. the number of steps needed for all vertices have the same colour.
We show that for a given initial colouring of the vertices, the expected time to reach consensus, depends strongly on the underlying graph and the update rule (push, pull, oblivious).

\ignore{
For oblivious voting on connected $n$-vertex graphs  and starting from an initial half red, half blue colouring, the expected time to consensus is $n^2/4$ independent of the underlying graph.

Let  $\E T$ be the expected time to finish voting.
For the push and pull protocols and the half--half initial colouring we have the following results.
For the complete graph $K_n$, the push protocol has $\E T= \Theta (n \log n)$, whereas the pull protocol has $\E T=\Theta(2^n)$. For the
cycle $C_n$ all three protocols have $\E T= \Theta(n^2)$. For the star graph however, the pull protocol has  $\E T=O(n^2)$, whereas the push protocol is slower with $\E T = \Theta(n^2 \log n)$. For the double star (two stars joined at the central vertices), for push we have  $\E T=\Om(2^{n/5})$ and for pull $\E T=O(n^4)$. Finally, for the barbell graph (two cliques of equal size joined by an edge) both push and pull have $\E T = \Om(2^{n/10})$.

The wide variation in $\E T$ for the pull protocol is to be contrasted with the well known model of synchronous pull voting, for which $\E T=O(n^3)$ on any connected graph, and $\E T = O(n)$ on many classes of expanders.
}

\end{abstract}


\section{Introduction}

The  process of reaching  consensus in a graph by means of local interactions is
known as voting. It  is an abstraction of human behavior, and can be implemented in distributed computer networks.
As a consequence  voting processes  have been widely studied.

In the simplest case each vertex has a colour (e.g. red, blue etc), and  neighbouring vertices
interact pairwise in a fixed way to update their colours. After this  interaction both vertices  have the same colour.
In randomized  voting,  three basic ways to make an update  are:
\begin{quote}
{\em Push:} Pick a random  vertex and push its colour to a random  neighbour.\\
{\em Pull:} Pick a random  vertex and pull the colour of a random  neighbour.\\
{\em Oblivious:} Pick a  random edge and push the colour of one randomly chosen endpoint to the other one.
\end{quote}
In the case of  asynchronous voting, all three methods are well defined. For synchronous voting
the push and oblivious processes are not well defined, as more than one colour could be pushed to a vertex at a given step.

A common discrete voting model is {\em randomized synchronous pull voting}.
In this model, at each step, each vertex synchronously adopts the opinion of a randomly chosen neighbour.
The  model has been extensively studied.
Hassin and Peleg~\cite{HassinPeleg-InfComp2001}
and Nakata { et al.}~\cite{Nakata_etal_1999}
proved that on  connected non-bipartite graphs the probability a given opinion $A$ wins is $d(A)/2m$ where $d(A)$ is the degree of the vertices initially holding opinion $A$, and $m$ is the number of edges.
For the time to consensus,
if the colours of the vertices are all initially distinct,
 the process takes $\Th(n)$ expected steps to reach consensus
on many classes of expander graphs on $n$ vertices.
This is proved for the complete graph  $K_n$ by Aldous and Fill  \cite{AlFi}), and for   $r$-regular random graphs  by Cooper, Frieze and Radzik \cite{CFR}. Results for general graphs based on the eigenvalue gap and variance of the degree sequence are given by Cooper et al. in \cite{CEHR-SIAM2013}.
They
 find an expected consensus time of $O(n/(\nu(1-\lambda)))$, where $n$ is the number of vertices of $G$, and $\l$ is the absolute value of the second eigenvalue of the transition matrix. The parameter $\nu$ measures the regularity of the degree sequence, and ranges from $1$ for regular graphs to $\Theta(n)$ for the star graph.
It is given by $\nu= \sum_{v\in V} d^2(v)/(d^2n)$, $d(v)$ is the degree of vertex $v$, and $d=d_{\text{ave}}=2m/n$ is the average degree.
For regular graphs, the result of \cite{CEHR-SIAM2013} achieves an upper bound of $O(n^3)$ in the worst case.  Using a different approach, Berenbrink et al. \cite{BGKT}  proved a consensus time of $O((d_{\text{ave}}/d_{\min})(n/\Phi))$. Here $d_{\text{ave}}, \;d_{\min}$ are the average and minimum degrees respectively. $\Phi$ is the graph conductance,
 $\Phi = \min_{S \subset V(G)} \frac{E(S:S^c)}{\min\{d(S),d(S^c)\}}$, where
 $E(S:S^c)$ are the edges between $S$ and $S^c$, and $S \ne \es,V$.

Much of the analysis of {\em asynchronous pull voting} has been made in the continuous-time model, where edges or vertices have exponential waiting times between events. An example is the  work by Cox \cite{Cox}  for toroidal grids. For detailed coverage see Liggett \cite{Lig}. More recently Oliveira \cite{Oliv} shows that the expected consensus time is $O(H_{\max})$, where
$H_{\max} = \max_{v,u \in V} H(v,u)$ and $H(v,u)$ is the expected hitting time of $u$ by a random walk starting at vertex $v$.
Asynchronous pull voting is less studied in a discrete setting. It was shown in \cite{CR} that the expected time to consensus for asynchronous pull voting is
\begin{equation}\label{Pull1}
\E T =O(n m/d_{\min}\Phi),
\end{equation}
where $m$ is the number of edges, $d_{\min}$ is minimum degree and $\Phi$ is graph conductance. Thus $\E T= O(n^5)$ for any connected graph, and
$O(n^2)$ for  regular expanders.

In this paper we consider a different asynchronous voting process, {\em discordant voting},  which can be described as follows. Initially each vertex holds one of two opinions, red or blue.  Neighbouring vertices of different colours, interact pairwise along a discordant edge. We reserve the term {asynchronous voting} for the ordinary case discussed previously.
This paper is a fundamental study of the expected time to consensus in discordant voting.
We find the  performance of the discordant voting process varies considerably both with the structure of the underlying graph, and the protocol used (push, pull, oblivious) and sometimes  in a quite counter-intuitive way (see Table \ref{Compare-results}). This behavior  is in stark contrast to that of the ordinary asynchronous case.

Discordant voting  originated in the complex networks community as a model of social evolution (see e.g. \cite{Sayama}, \cite{Sayama1}).
 The general version of the  model allows for {\em rewiring}. The interacting vertices can break  edges joining them and reconnect elsewhere. This serves as a model of social behavior  in which vertices  either change their opinion or their friends.

Holme and Newman \cite{HN} investigated discordant  voting  as
a model of  a self-organizing network which restructures based on the acceptance or rejection of differing opinions among social groups.
At each step, a random discordant edge $ uv$ is selected, and an endpoint $x \in \{u,v\}$ chosen with probability $1/2$.
With probability $1-\a$ the opinion of $x$ is pushed to the other endpoint $y$, and with probability $\a$, vertex $y$ breaks the edge and rewires to a random vertex with the same opinion as itself. Simulations suggest the existence of threshold behavior in $\a$. This was investigated further by
Durrett et al. \cite{Durrett}
for sparse random graphs of constant average degree 4 (i.e. $G(n,4/n)$).
The paper studies two rewiring strategies, rewire-to-random, and rewire-to-same, and finds experimental evidence of a phase transition  in both cases.
Basu and Sly \cite{BS} made a formal analysis of rewiring for Erdos-Renyi graphs $G(n,1/2)$ with $1-\a=\b/n$, $\b>0$ constant.
They found that for either strategy, if $\b$ is sufficiently small the network quickly disconnects maintaining the initial proportions. As $\b$ increases
the minority proportion decreases,  and in rewire-to-random a positive fraction of both opinions survive. A subsequent paper by Durrett et al. \cite{Durrett2} examines the rewiring phase transitions for the intermediate case of thick graphs $G(n,1/n^a)$ where $0<a<1$.

Although discordant voting seems a natural model of local interaction, its behavior is not well understood even in the simplest cases. Moreover, the analysis of rewiring is highly problematic. Firstly  there is no natural model for the space of random graphs derived from the rewiring. Secondly the voting and rewiring interactions condition the degree sequence in a way which makes subsequent analysis difficult.

In this paper we assume there is no rewiring, and evaluate the performance of discordant voting as a function of the graph structure.
Discordant voting always chooses an edge  between the opposing red and blue sets, so intuitively it should finish faster than  ordinary  asynchronous voting  which ignores this discordancy information.

Perhaps surprisingly,
for discordant voting using the oblivious protocol, the expected time to
consensus  is the same for any connected $n$--vertex graph. It is independent of graph structure and of the number of edges, and depends only on
 the initial number of vertices of each colour (red, blue).
Whichever discordant edge is chosen,  the number of blue vertices in the graph increases (resp. decreases)  by one with probability $1/2$ at each step.
This is  equivalent to an unbiased   random walk on the line $(0,1,...,n)$ with absorbing barriers (see Feller \cite[XIV.3]{Feller}).

\begin{remark}\label{Obliv}
{\em Oblivious protocol.}
Let $T$ be the time to consensus in the two-party asynchronous discordant voting
process starting from any initial coloring with $R(0)=r, B(0) =n-r$ red and blue vertices respectively.
  For any connected $n$ vertex graph, $\E T(\text{Oblivious})= r(n-r)$.
\end{remark}
Starting with an  equal number of red and blue vertices  the oblivious protocol takes  $\E T \sim n^2/4$ steps for any connected graph. For ordinary asynchronous voting, the performance of the oblivious protocol can also depend on the number of edges $m$. In the worst case expected wait to hit the last red-blue edge is $m$,  so the
ordinary case  takes $\E T =O(m n^2)$ steps.

In  contrast to the oblivious case, discordant  push and pull protocols can exhibit very different expected times to consensus, which depend strongly on the underlying graph in question.

\begin{theorem}
 \label{Kn}
Let $T$ be the time to consensus of the asynchronous discordant voting
process starting from any initial coloring with an equal number of red and blue vertices $R=B=n/2$.
For the complete graph $K_n$,
$\E T(\text{Push}) = \Th ( n \log n)$,  and $\E T(\text{Pull})= \Th(2^n)$.
\end{theorem}

Thus for the complete graph $K_n$
the different protocols give  very different expected completion times, which vary from  $\Th ( n \log n)$ for push, to $\Th(n^2)$ for oblivious, to  $\Th(2^n)$ for pull. On the basis of this evidence, our initial view was that  there should be a meta-theorem of the \lq push is faster than oblivious, oblivious is faster than pull' type.  Intuitively, this is supported by the following argument. Suppose  red ($R$) is the larger colour class.
 Choosing a  discordant vertex uniformly at random, favors the selection of the
 larger class. In the push process, red vertices push their opinion more often, which tends
 to increase the size of $R$.
  Conversely, the pull process  tends to re-balance the set sizes. If $R$ is larger, it is recoloured more often.

For the cycle $C_n$, we prove that all three protocols have similar expected time to consensus; a result which is consistent with the above meta-theorem.
\begin{theorem}
 \label{Cn}
Let $T$ be the time to consensus of the asynchronous discordant voting
process starting from any initial coloring with an equal number of red and blue vertices $R=B=n/2$.
For the cycle $C_n$, the Push, Pull and Oblivious protocols have 
$\E T=\Th(n^2)$.
\end{theorem}
At this point we are left with a  choice. Either to produce evidence for a  relationship of the form   $\E T(\text{Push}) =O( \E T(\text{Pull}))$ for general graphs, or to refute it. Mossel and  Roch \cite{Mossel} found slow convergence of the iterated prisoners dilemma problem (IPD) on caterpillar trees. Intuitively push voting is aggressive, whereas pull voting is altruistic, and thus similar to cooperation in the IPD. Motivated by this, we found simple counter examples, namely the star graph $S_n$ and  the double star $S_n^*$.

\begin{theorem}
 \label{Sn}
Let $T$ be the time to consensus in the two-party asynchronous discordant voting
process starting from any initial coloring with an equal number of red and blue vertices $R=B=n/2$.

For the star graph $S_n$,
$\E T(\text{Push}) = \Th ( n^2 \log n)$,  and $\E T(\text{Pull})= O(n^2)$.

For the double star $S^*_n$ with the initial colouring of Figure \ref{double:fig01},
$\E T(\text{Push}) = \Om(2^{n/5})$,  and $\E T(\text{Pull})= O(n^4)$.
\end{theorem}

\begin{figure}[H]
\centerline{%
    {\begin{tikzpicture}
    \foreach \x in {4,...,14}
    {\draw (20*\x:2cm) node[rnode] (\x1) {};}
    \draw (0,0) node[rnode] (c1) {} (0.4,0.25) node {\large$c_1$};
    \foreach \x in {4,...,14}
    {\draw (c1)--(\x1) ;}
    \draw (-2.5,0) node {{\large$S_1$}} ;
    \begin{scope}[xshift=4cm,rotate=180]
    \foreach \x in {4,...,14}
    {\draw (20*\x:2cm) node[bnode] (\x2) {};}
    \draw (0,0) node[bnode] (c2) {} (0.4,-0.25) node {\large$c_2$} ;
    \foreach \x in {4,...,14}
    {\draw (c2)--(\x2) ;}
    \draw (-2.5,0) node {{\large$S_2$}} ;
    \end{scope}
    \draw (c1)--(c2);
  \end{tikzpicture}}}
  \caption{Double star $S^*$ with half of the vertices coloured red and half coloured blue.
  }\label{double:fig01}
\end{figure}
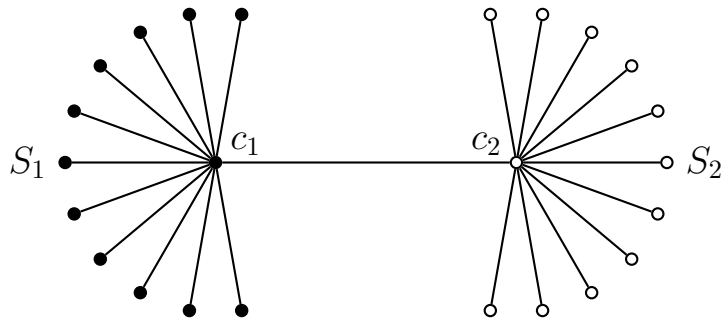

\ignore{
 In Section \ref{Sim} we give  simulation results for $C_n$ and $S_n$
for small graphs ($n \le 500$). Fig. \ref{star-sim} indicates a rapid convergence to $\E T$. For the  star graph $S_n$,  the experiments
show a clear separation in $\E T$ for  the three protocols.
}

At this point little remains of the possibility of  a meta-theorem except a vague hope that at least one of the push and pull protocols always has polynomial time to consensus. However, this is disproved by the example of the {\em barbell graph}, which consists of two cliques of size $n/2$ joined by a single edge.
\begin{theorem}
\label{Bn}
Let $T$ be the time to consensus in the two-party asynchronous discordant voting
process starting from any initial coloring with an equal number of red and blue vertices $R=B=n/2$.

For the barbell graph,
$\E T(\text{Push}) = \Om (2^{2n/5} )$,  and $\E T(\text{Pull})= \Th(2^n)$.
\end{theorem}

A summary of these results is given in the table below.

\begin{figure}[H]
\begin{center}
\begin{tabular}{|l|c|c|c||c|c|c|}
\hline
&\multicolumn{3}{|c||}{Discordant voting}&\multicolumn{3}{|c|}{Asynchronous voting}\\
\cline{2-7}
& Push & Pull &Oblivious&Push & Pull&Oblivious\\
\hline\hline
Complete graph $K_n$ & $\Th(n \log n)$ & $\Th(2^n)$&\multirow{5}{*}{ $n^2/4$}&
$O(n^2)$&$O(n^2)$&$O(n^4)$\\
\cline{1-3} \cline{5-7}
Cycle $C_n$ & $\Th(n^2)$&$\Th(n^2)$&&$O(n^2)$&$O(n^2)$&$O(n^3)$\\
\cline{1-3} \cline{5-7}
Star graph $S_n$& $\Th(n^2 \log n)$& $O(n^2)$&&$O(n^2)$&$O(n^2)$&$O(n^3)$\\
\cline{1-3}\cline{5-7}
Double star $S^*_n$&$\Om(2^{n/5})$&$O(n^4)$&&$O(n^3)$&$O(n^4)$&$O(n^3)$\\
\cline{1-3}\cline{5-7}
Barbell graph& $\Om(2^{n/10})$& $\Theta(2^{n/2})$&&$O(n^4)$&$O(n^4)$&$O(n^4)$\\
\hline
\end{tabular}
\end{center}
\caption{Comparison of expected time to consensus ($\E T$) for discordant and ordinary asynchronous voting protocols on connected $n$-vertex graphs, starting from $R=B=n/2$.}
\label{Compare-results}
\end{figure}

The column for ordinary asynchronous pull voting in Table \ref{Compare-results} follows from
\eqref{Pull1}. The column for ordinary asynchronous pull voting from $\E T=O(n^2 m)$ (see below Remark \ref{Obliv}). To complete the column for ordinary asynchronous push voting, we
used a result of \cite{CR}. For any graph $G=(V(G),E(G))$,
\begin{equation}\label{Push1}
\E T(\text{push})=O(1/\Psi(G)),
\end{equation}
 where
\begin{equation}\nonumber
\Psi(G) = \frac{2C(G)}{n d_{\max}}\;
\min_{S \subset V(G)} \frac{1}{\min\{J(S),J(S^c)\}}\;
\sum_{(v,w) \in E(S:S^c)}  \frac{1}{d(v)d(w)}.
\end{equation}
The expression is evaluated over sets $S \ne \es, V(G)$, and $d_{\max}$ is maximum degree, $C(G) = (\sum_{v \in V} 1/d(v))^{-1}$, $E(S:S^c)$ are the edges between $S$ and $S^c$,
and $J(S) = \sum_{v \in S} d(v)^{-1}$.
The parameter $\Psi$ does not seem related to the classical graph parameters,
but can be directly evaluated for the graphs we consider.
For regular graphs,
\[
\Psi=\frac{2}{n^2} \Phi,
\]
in which case $\E T =O(n^2/\Phi)$, which agrees with the asynchronous pull model in \eqref{Pull1}.

\ignore{
 {\bf Other results on discordant voting.}

The performance of discordant voting on random graphs mirrors that of
Theorem \ref{Kn} for the complete graph. The proofs require a level of detail
which seems excessive as compared to the simple examples given here, and will be given in \cite{CDFGnp}.
In a  different direction the paper \cite{CRKn} considers the range of behavior
of discordant voting on $K_n$ in the case where there is a mixture of $\a$ push and $1-\a$ pull voting. When $\a=1$ we have push voting, and when $\a=0$ pull voting. The case $\a=1/2$ corresponds to oblivious, and the entire range of $\E T$ from $\Th (n \log n)$ to $2^n$ can be found by judicious choice of $\a$.
}

\subsection*{Asynchronous discordant voting model }

We next give a formal definition of the discordant voting process.
 Given a graph $G=(V,E)$, with $n=|V|$. Each vertex $v\in  V$ is labelled with an \emph{opinion} $X(v)\in \{0,1\}$. We  call
$X$ a \emph{configuration} of opinions. We can think of the opinions as having colours; e.g. red (0)  and blue (1),  or  black (0) and white (1).
An edge $e=uv\in E$ is \emph{discordant} if $X(u)\neq X(v)$. Let $K(X)$ denote the set of  discordant edges at time $t$. A vertex $v$ is discordant if it is incident with any discordant edge, and $D(X)$ will denote the set of discordant vertices in $X$. We consider three random update rules for opinions $X_t$ at time $t$.
\begin{description}\label{descrip}
  \item[Push:] Choose $v_t\in D(X_t)$, uniformly at random, and a discordant neighbour $u_t$ of $v_t$ uniformly at random. Let $X_{t+1}(u_t)\gets X_t(v_t)$, and $X_{t+1}(w)\gets X_t(w)$ otherwise.
  \item[Pull:] Choose $v_t\in D(X_t)$, uniformly at random, and a discordant neighbour $u_t$ of $v_t$ uniformly at random. Let $X_{t+1}(v_t)\gets X_t(u_t)$, and $X_{t+1}(w)\gets X_t(w)$ otherwise.
  \item[Oblivious:] Choose $\{u_t,v_t\}\in K(X_t)$ uniformly at random. With probability \nicefrac12,  $X_{t+1}(v_t)\gets X_t(u_t)$, with probability \nicefrac12,  $X_{t+1}(u_t)\gets X_t(v_t)$,  and $X_{t+1}(w)\gets X_t(w)$ otherwise.
\end{description}
These three processes are Markov chains on the configurations in $G$, in which the opinion of exactly one vertex is changed at each step.
\ignore{
The three update rules select this \emph{change} vertex in different ways. For example, the push rule selects an {\em active} discordant vertex $u$ and the change vertex is a discordant neighbour $v$.
Once the change vertex is chosen it simply flips its opinion from $0$ to $1$ or vice versa.
}
Assuming $G$ is connected, there are two absorbing states, when $X(v)=0$ for all $v\in V$, or $X(v)=1$ for all $v\in V$, where no discordant vertices exist. When the process reaches either of these states, we say that is has converged. Let $T$ be the step at which convergence occurs. Our object of study is $\E T$.

\subsection*{Structure of the paper.}

A major obstacle in the analysis discordant voting, is that the effect of recoloring a vertex is not always monotone.
For each of the graphs studied, the way to bound $\E T$ differs.
The proof of the pull voting result for the cycle $C_n$ in particular, is somewhat delicate, and  requires an analysis of the optimum of a linear program based on a potential function.

The general proof methodology is to map the process to a biased random walk on the line $0,...,n$. In Section \ref{BDC} we prove results for a Birth-and-Death chain which we call the Push chain. This chain can be coupled with many aspects of the discordant voting process.
We then prove Theorems \ref{Kn}, \ref{Cn}, \ref{Sn} and \ref{Bn}  in that order.


\section{Birth-and-Death chains}\label{BDC}
A Markov chain $(X_t)_{t \geq 0}$ is said to be a Birth-and-Death chain on state space $S=\{0,\ldots, N\}$ if given $X_t =i$ then the possible values of $X_{t+1}$ are $i+1, i$ or $i-1$ with probability $p_i,$ $r_i$ and $q_i$ respectively. Note that $q_0 = p_N = 0$.
In this section we assume that $r_i=0$, $p_0=1$, $q_N=1$, $p_i > 0$ for $i \in \{0,\ldots, N-1\}$ and $q_i > 0$ for $i \in \{1,\ldots, N\}$. We denote $\E_i Y$ the expected value of random variable $Y$ when the chain starts in $i$ (i.e., $X_0 = i$). Finally, we define the (random) hitting time of  state $i$ as
$T_i = \min \{t \geq 0: X_t = i\}.$

We summarize  the results we require on Birth-and-Death chains (see Peres, Levin and Wilmer \cite[2.5]{LPW}).

Say that a probability distribution $\pi$ satisfies the detailed balance equations, if
\begin{equation}\label{eqn:balanced}
\pi(i) P(i,j) = \pi(j)P(j,i) \text{, for all $i,j \in S$}.
\end{equation}
Birth-and-Death chains with $p_i=P(i,i+1), q_i=P(i,i-1)$  can be shown to satisfy the detailed balance equations.
It follows from this, (see e.g. \cite{LPW})  that
\begin{equation}\label{eq:localExpHit}
\E_{i-1}{T_i} = \frac{1}{q_i\pi(i)}\sum_{k=0}^{i-1}\pi(k)
\end{equation}
An equivalent formulation (see \cite{LPW}) is $\E_0T_1=1/p_0=1$ and in general
\beq{ETi}
E_{i-1}T_{i} =  \sum_{k=0}^{i-1} \frac{1}{p_k}\frac{q_{k+1}\cdots q_{i-1}}{ p_{k+1} \cdots p_{i-1}},\quad
\text{ for } i \in \{1,\ldots, N\}.
\eeq
In writing this expression we follow the convention that if $k=i-1$ then
$\frac{q_{k+1}\cdots q_{i-1}}{p_{k+1}\cdots p_{i-1}}=1$ so that the last term is
$1/p_{i-1}$. Note also that the final index $k$ on $p_k$ is $k=N-1$, i.e. we never divide by $p_N=0$.

Starting from state 0, let $T_M$ be the number of transitions needed to reach state $M$ for the first time. For any $M \le N$, we have that $\E_0T_M= \sum_{i=1}^{M} \E_{i-1} T_{i}$. For example, $\E_0 T_1 = \frac{1}{p_0}=1$ and $\E_0 T_2 = 1+\frac{1}{p_1}+\frac{q_1}{p_0p_1}$ etc.
Thus, for $M \ge 1$
\begin{equation}\label{ETM}
\E_0T_M = \sum_{i=1}^{M} \E_{i-1}T_{i} =
\sum_{i=1}^{M}\sum_{k=0}^{i-1} \frac{1}{p_k }\prod_{j=k+1}^{i-1}\frac{q_j}{p_j}.
\end{equation}

We define two Birth-and-Death chains which feature  in our analysis.
The chains have states $\{0,1,...,i, ..., N\}$ where $N=n/2$ (assume $n \ge 2$ even). The transition probabilities from state $i$ given by $P(i,i+1)$, $Q(i,i+1)=1-P(i,i+1)$.
We refer to these chain as the push chain, and pull chain respectively.

\paragraph{Push Chain.}
Let $Z_t$ be the state occupied by the push chain at step $t \ge 0$.
Let $\d \in \{-1,0,+1\}$ be fixed. When applying results for the push chain
in our proofs, we will state the value of $\d$ we use.
The transition probability $p_i=P(i,i+1)$ from $Z_t=i$,
is given by
\begin{equation}\label{eqn:p_idef}
p_i =
\begin{cases}
1, & \text{if }i = 0 \\
1/2+i/n +\d/n, & \text{if }i \in \{1, \ldots, n/2-1\} \\
0 , & \text{if }i = n/2
\end{cases}.
\end{equation}

\paragraph{Pull Chain.}
Let $\ol Z_t$ be the state occupied by the pull chain at step $t \ge 0$.
Given that $\ol Z_t=i$, the transition probability $\ol p_i=\ol P(i,i+1)$ is given by
\begin{equation}\label{eqn:pu_idef}
\ol p_i =
\begin{cases}
1, & \text{if }i = 0 \\
1/2-i/n -\d/n, & \text{if }i \in \{1, \ldots, n/2-1\} \\
0 , & \text{if }i = n/2
\end{cases}
.
\end{equation}
For $1 \le i \le N-1$ the pull chain  is the push chain with the probabilities reversed, i.e. $\ol p_i=q_i$.

\subsection*{Push Chain: Bounds on hitting time}
{\bf Push Chain: Upper bound on hitting time.}

\begin{lemma}\label{slow}
For any $M \le N$ , let $\E_0T_M$ be the expected hitting time
of $M$ in the push chain $Z_t$ starting from state 0.
Then
\[
\E_0 T_M \le 2 N \log M +O(1).
\]
\end{lemma}

\begin{proof}
Using \eqref{ETM} and recalling the notational convention given below \eqref{ETi}
we can  change the order of summation to give
\begin{equation}\label{eqn:basicSum}
E_0T_M  = \sum_{k=0}^{M-1} \sum_{i=k+1}^{M} \frac{1}{p_k }\frac{q_{k+1}\cdots q_{i-1}}{p_{k+1} \cdots p_{i-1}}=
\frac{1}{p_{M-1}}+ \sum_{k=0}^{M-2} \sum_{i=k+1}^{M-1} \frac{1}{p_k }\frac{q_{k+1}\cdots q_{i-1}}{p_{k+1} \cdots p_{i-1}}.
\end{equation}
Using \eqref{eqn:p_idef}, we see that for $1 \le k \le N-2$ we see that $q_k/p_k \ge q_{k+1}/p_{k+1}$, $q_1/p_1 \le 1$, and for $2 \le k \le N-1$ that  ${q_{k}}/{p_{k}} < 1$.
As $p_0=1$, we upper bound  $E_0 T_M$ by
\beq{flick}
E_0 T_M \leq {M}+\frac{1}{p_{M-1}}+\sum_{k=1}^{M-2} \frac{1}{p_k }\sum_{i=k+1}^{M-1} \left(\frac{q_{k+1}}{p_{k+1}}\right)^{i-k-1},
\eeq
and
\beq{EOTMval}
 \sum_{k=1}^{M-2} \frac{1}{p_k }\sum_{\ell=0}^{\infty} \left(\frac{q_{k+1}}{p_{k+1}}\right)^\ell =
\sum_{k=1}^{M-2} \frac{1}{p_k} \frac{1}{1-\frac{q_{k+1}}{p_{k+1}}}=
\sum_{k=1}^{M-2}\frac{p_{k+1}}{p_k}\frac{1}{p_{k+1}-q_{k+1}}.
\eeq
As $q_k=1-p_k$, $p_k-q_k=2p_k-1>0$ for all $k \in \{2,\ldots, N-1\}$,
then $\frac{1}{p_k-q_k} =\frac{N}{k+\d}$. For all $k \in \{1, \ldots, N-2\}$
we have $\frac{p_{k+1}}{p_k} \leq 2$.
Using \eqref{flick} with the upper bounds given in \eqref{EOTMval},  we obtain the required conclusion.
\end{proof}

{\bf Push Chain: Lower bound on hitting time.}

\begin{lemma}\label{qwik}
Let $\d=0$ in \eqref{eqn:p_idef}.
 Let $\E_0T_M$ be the expected hitting time
of $M$ in the push chain $Z_t$ starting from state 0.
There exists a constant $C$ such that, for any $\sqrt{N} \le M =o(N^{3/4})$,
\[
\E_0 T_M \ge C (N \log M/\sqrt{N} +\sqrt{N}).
\]
\end{lemma}
\begin{proof}
For $0<x <1$,
\[
\frac{1-x}{1+x}= \exp \left\{-2\brac{x+ \frac{x^3}{3}+ \cdots + \frac{x^{2\ell+1}}{2\ell+1}+\cdots}\right\}.
\]
Thus with $N=n/2$
\begin{flalign}
\prod_{j=k+1}^{i-1} \frac{q_j}{p_j} &= \prod_{j=k+1}^{i-1} \frac{1-j/N}{1+j/N}
\label{lab1}\\
&= \exp\left\{ -2 \brac{\sum_{j=k+1}^{i-1} \frac jN + \sum \frac{ (j/N)^3}{3}+ \cdots
+\sum \frac{(j/N)^{2\ell+1}}{2\ell +1}+\cdots}\right\}\nonumber\\
&=\exp \{-2 \F\},\label{lab2}
\end{flalign}
say. If $f(s)$ is non-negative and monotone increasing, then
$\sum_{s=k+1}^{i-1} f(s) \le \int_{k}^i f(s)\,ds$. Thus,
the sum of terms in $(j/N)^3$ and above in $\F$ can be bounded above by
\begin{flalign*}
\sum_{\ell \ge 1}\sum_{j=k+1}^{i-1} \frac{(j/N)^{2 \ell +1}}{2 \ell+1}
&\le \sum_{\ell \ge 1}\frac{1}{(2\ell+1)N^{2\ell+1}} \int_k^i x^k dx\\
\le  \sum_{\ell \ge 1}\frac{1}{(2\ell+1)N^{2\ell+1}} \cdot \frac{i^{2\ell+2}}{2\ell+2}\\
&= O\bfrac{i^4}{N^3}\sum\frac{1}{(2\ell+1)(2\ell+2)}=O\bfrac{i^4}{N^3}.
\end{flalign*}
Thus, using our assumption that $M=o(N^{3/4})$,
\[
\F= \frac{i(i-1)}{2N} -\frac{k(k+1)}{2N} + O \brac{ \frac{i^4}{N^3}}= \frac{i^2}{2N}-\frac{k^2}{2N}- \frac{i+k}{2N}- o(1).
\]
Replacing $\F$ in \eqref{lab2} with the upper bound given above,  gives a lower bound on the term \eqref{lab2} in \eqref{ETM}. Thus
\beq{Annoying}
\E_0T_M\ge (1-o(1))  \sum_{i=0}^M  \sum_{k=0}^{i-1} \frac {1}{p_k} \exp\brac{-\frac{i^2}{N}}\exp\brac{\frac{k^2}{N}}.
\eeq
For $i \le M$ the last term on the righthand side of \eqref{Annoying} is bounded below
by a positive constant.
Let
\beq{ssi}
\s(i) = \sum_{k=0}^{i-1} \exp \bfrac{k^2}{N}.
\eeq
Let $\b=(1/2) \log 2 \approx 0.34$.
We claim that, if $i \ge \sqrt{N}$ then
\beq{Si}
\s(i) \ge \frac{\b N}{2i} e^{i^2/N}.
\eeq
Let $a=\b N/i$ then for $i\ge \sqrt N$, $i-a>0$. For $k \ge i-a$
\[
\frac{k^2}{N} \ge \frac{i^2}{N}-\frac{2ia}{N}+\frac{a^2}{N}=
\frac{i^2}{N}-\frac{2i}{N}\b\frac{N}{i}+\b\frac{N}{i^2}
\ge\frac{i^2}{N}-2\b.
\]
If $k \ge i-a$, then $\exp \{k^2/N\} \ge \frac 12 \exp \{i^2/N\}$.
As there are at least $a$ such values of $k$, it follows that $\s(i) \ge \b N/2i e^{i^2/N}$.

Let $ \sqrt N \le i \le M=o(N^{3/4})$.
 Replace \eqref{ssi} in\eqref{Annoying} with \eqref{Si}. Noting that $p_0=1$ and  for $1\le k \le M$, $p_k\sim 1/2$, we can assume $(1-o(1))/p_k \ge 1/2$ to give
\[
\E_0 T_M \ge  \sum_{i< \sqrt{N}} \frac{e^{-1}}{2}+ \sum_{i=\sqrt{N}}^M \frac{\b N}{2 i}\ge \sqrt{N}/6+
 \frac{\b N}{3} \log \frac{M}{\sqrt N}.
\]
\ignore{
\[
\sum_{i=\sqrt{N}}^M\sum_{k=0}^{i-1} \frac{1}{p_k}\prod_{j=k+1}^{i-1}
 \frac{q_j}{p_j} \ge
  \sum_{k=0}^{i-1} \Th(1)\;\s(i)\; e^{-i^2/N}
 \ge \Th \bfrac Ni.
\]
Thus
\[
\E_0T_M \ge \sqrt{N}+ \Th(1) \sum_{i= \sqrt N}^M \frac Ni \ge \Th(N) \log \frac{M}{\sqrt N}.
\]
}
\end{proof}


\section{ Voting on the complete graph $K_n$.} \label{KN}
For the complete graph $K_n$, the probability $B$ increases at a given step is $B(t)/n$, whereas in the pull process it is $R(t)/n=1-B(t)/n$.
The chain  defined by $Y_t = \max\{R(t), B(t)\}-n/2$ is a Birth-and-Death chain.
We  study the time that takes $Y_t$ to reach $N=n/2$ starting from 0.

{\bf Theorem \ref{Kn}: Push process.}
For the push model, the process $Y_t$ is
identical to the push chain $Z_t$  with transitions given
by \eqref{eqn:p_idef} with $\d=0$. This was analysed Section \ref{BDC}.

{\bf Theorem \ref{Kn}: Pull process.}
For the pull model, the process $Y_t$ is
identical to the pull chain $\ol Z_t$  with transitions given
by \eqref{eqn:pu_idef} with $\d=0$

For the pull model, the process $Y_t$ is identical to the pull chain $\ol Z_t$  with transitions given by \eqref{eqn:pu_idef}. To begin with, observe that  $w_k = \binom{n}{N+k},k=0,1,\ldots,N$ satisfies the detailed balance equation \eqref{eqn:balanced}. Hence we have $\p(k)=w_k/W$, where $W=w_0+w_1+\cdots+w_N$.

It follows from \eqref{eq:localExpHit} that
$$\E_{i-1}T_i = \frac{2n}{n+2i}\cdot\frac{1}{\binom{n}{N+i}}\cdot\sum_{k=0}^{i-1}\binom{n}{N+k}.$$
Putting $i=N$ we have
\begin{equation}\label{Knfrom1}
\E_{N-1}T_N=\sum_{k=0}^{N-1}\binom{n}{N+k}= \frac 12 \brac{2^{n}-2+  \binom{n}{N}}
=\Omega(2^n).
\end{equation}
On the other hand, an upper bound
$$\sum_{i=1}^N\E_{i-1}T_i\leq 2\cdot
2^n\cdot \sum_{i=1}^N\frac{1}{\binom{n}{N+i}}=O(2^n),$$
follows from a result of Sury \cite{Sury}, that
$$\sum_{i=1}^{N} \frac{1}{\binom{n}{N+k}} = \frac{n+1}{2^n}\sum_{i=0}^n \frac{2^i}{i+1} = O(1).$$


\section{Voting on the cycle}
An $n$-cycle $G$, with $V=[n]$, has $E=\{ (i,\,i+1):i\in[n]\}$, where we identify vertex $n+i$ with vertex $i$. See Fig.~\ref{cycle:fig01}.

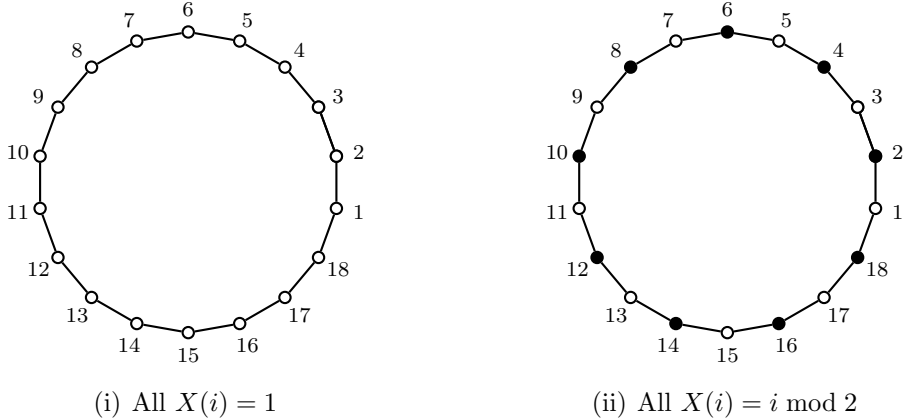
\begin{figure}[H]
\centerline{%
    \subfigure[All $X(i)=1$]{\begin{tikzpicture}[font=\scriptsize]
    \foreach \x in {0,2,...,18}
    {\draw (20*\x+10:2cm) node[bnode] (\x) {};}
     \foreach \x in {1,3,...,19}
    {\draw (20*\x+10:2cm) node[bnode] (\x) {};}
     \foreach \x in {0,1,2,...,18}
    {\pgfmathtruncatemacro{\y}{\x + 1} \draw (\x) to (\y) ;};
    \foreach \x in {1,2,...,18}
    {\draw (20*\x-30:2.3cm) node {\x};}
  \end{tikzpicture}}\hspace*{2cm}
  \subfigure[All $X(i)=i \bmod 2$]{\begin{tikzpicture}[font=\scriptsize]
    \foreach \x in {0,2,...,18}
    {\draw (20*\x+10:2cm) node[rnode] (\x) {};}
     \foreach \x in {1,3,...,19}
    {\draw (20*\x+10:2cm) node[bnode] (\x) {};}
     \foreach \x in {0,1,2,...,18}
    {\pgfmathtruncatemacro{\y}{\x + 1} \draw (\x) to (\y) ;}
    \foreach \x in {1,2,...,18}
    {\draw (20*\x-30:2.3cm) node {\x};}
  \end{tikzpicture}}}
  \caption{Cycle with $n=18$}\label{cycle:fig01}
\end{figure}

Let $X=X(t)$ denote the (configuration of opinions) of  the voting process at time $t$,
Let $K(X)$ denote the set of  discordant edges of $X$
and let $k(X)=|K|$. Let $D(X)$  denote the set of discordant vertices in $X$.

We say $i+1$, $i+2,\ldots,j$ is a \emph{run} of length $(j-i)$ $(1\leq j-i<n)$ if $X(i)\neq X(i+1)=X(i+2)=\cdots = X(j)\neq X(j+1)$. A \emph{singleton} is a run of length 1, a single vertex. These vertices require special treatment, since they lie in two discordant edges. Note that the number of runs, $k(X)$, in $X$ is equal to the number of discordant edges. Also $k$ is even, since red and blue runs must alternate, so we will write $r(X)=\frac12k(X)$, and $k_0=2r_0=k(X_0)$. Thus $r(X)$ is the number of paths of a given colour. Then $T$ is the first $t$ for which $k(X_t)=r(X_t)=0$, (a cycle is not a path).

Let the $k$ runs in $X$ have lengths $\ell_1,\ell_2,\ldots,\ell_k$ respectively, and let $s(X)$ denote the number of singletons. Clearly $\sum_{i=1}^k\ell_i=n$, and there are $\kappa=2k-s$ discordant vertices, so $k\leq\kappa\leq 2k$.

We wish to determine the convergence time $T$ for an arbitrary configuration $X_0$ of the push or pull process to reach an absorbing state $X_T$ with $X_T(i)=X_T(1)$ $(i\in[n])$. In these processes, the run lengths behave rather like symmetric random walks on the line. However, an analysis using classical random walk techniques~\cite{Feller} seems problematic. There are two main difficulties. Firstly, the $k$ ``walks'' (run lengths) are correlated. If a run is long, the adjacent runs are likely to be shorter, and vice versa. Secondly, when the change vertex is a singleton, the lengths of three adjacent runs are combined, so three walks suddenly merge into one. One of the three runs is a singleton, but the other two may have arbitrary lengths.

Therefore, we will use the random walk view only to give a lower bound on the convergence time. For the upper bound, we use a different approach. We will define a \emph{potential function}
\[ \psi(X)\,=\,\sum_{i=1}^k \sqrt{\ell_i}\,,\]
where $\psi(X)=0$ if and only if $k(X)=0$. The important feature of $\psi$ is that it is a separable and strictly concave function of the $\ell_i$ ($i\in[k]$). Almost any other function with these properties would give similar results.
\begin{lemma}\label{cycle:lem01}
For any configuration $X$ on the $n$-cycle with $k$ runs, $\psi(X) \leq \sqrt{kn}$.
\end{lemma}
\begin{proof}
If $k=0$, this is clearly true. Otherwise, if $k\geq 2$, by concavity we have $\psi(X)/k=\frac1k\sum_{i=1}^k\sqrt{\ell_i}\,\leq\,\sqrt{\frac1k\sum_{i=1}^k\ell_i}\,=\,\sqrt{n/k}$,
so $\psi(X)\,\leq\,\sqrt{kn}$.
\end{proof}
Observe that $k(X_{t+1})=k(X_t)$ at step $t$ of either the push or pull process, unless  the change vertex is a singleton, in which case we may have $k(X_{t+1})=k(X_t)-2$. Thus $\{t: k(X_t)=2r\}$ is an  interval $[t_r,t_{r-1})$, which we  will call \emph{phase $r$} of the process.

Let $v_t=v\in D(X_t)$ be the active vertex, i.e. the vertex selected to push in the push rule, or pull in the pull rule. Let $\delta_v$ be the expected change in $\psi$, i.e.
\[ \delta_v\,=\,\E[\,\psi(X_{t+1}) - \psi(X_t)\mid v_t=v\,].\]
If there are $\kappa=2k-s$ discordant vertices,  the total expected change $\delta$  in $\psi$ is
\[ \Delta\,=\,\E[\,\psi(X_{t+1}) - \psi(X_t)]\, =\, \frac{1}{\kappa}\sum_{v\in D}\delta_v.\]
We will show that $\Delta$ is negative, so $\psi(X_t)$ is monotonically decreasing with $t$, in expectation.  Unfortunately we cannot simply bound $\delta_v$ for each $v\in D$, since it is possible to have $\delta_v>0$. Thus we will consider discordant \emph{edges}. We partition the set $K$ of discordant edges $uv$ into three subsets:
\begin{enumerate}[topsep=0pt,itemsep=0pt,label=(\Alph*)]
  \item\quad $A=\{uv : $ \mbox{$u$ and $v$ not singleton}$\}$\/;
  \item\quad $B=\{uv : $ \mbox{$u$ not singleton, $v$ singleton}$\}$\/;
  \item\quad $C=\{uv : $ \mbox{$u$ and $v$ both singleton}$\}$\/.
\end{enumerate}
See Fig.~\ref{cycle:fig02}, where $\ell_z$ is the length of the run containing discordant vertex $z$, for $z\in\{u,v,w,q\}$.

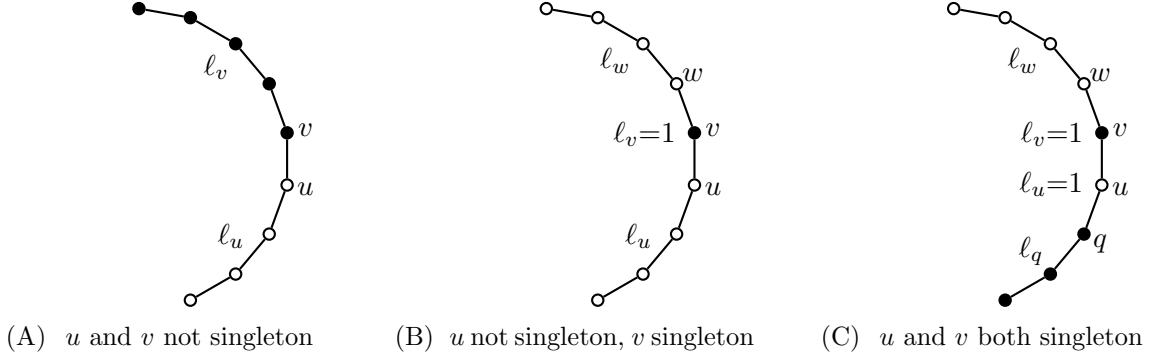
\begin{figure}[H]
\renewcommand{\thesubfigure}{(\Alph{subfigure}) }
\centerline{%
\subfigure[$u$ and $v$ not singleton]{\hspace*{7mm}\begin{tikzpicture}[font=\small]
    \draw  (2.75,0) coordinate (r) (-1.5,0) coordinate (l) ;
    \foreach \x in {-4,-3,-2,-1}
    {\draw (20*\x+10:2) node[bnode] (\x) {};};
    \foreach \x in {1,2,3,4}
    {\draw (20*\x+10:2) node[rnode] (\x) {};};
     \draw (10:2) node[rnode] (0) {};
     \foreach \x in {-4,-3,...,3}
    {\pgfmathtruncatemacro{\y}{\x + 1} \draw (\x) to (\y) ;};
    \draw (50:1.6) node {\black{$\ell_v$}}; \draw (-40:1.6) node {$\ell_u$};
    \draw (-10:2.25) node {$u$}; \draw (10:2.25) node {$v$};
  \end{tikzpicture}}\hspace*{5mm}
\subfigure[$u$ not singleton, $v$ singleton]{\hspace*{5mm}\begin{tikzpicture}[font=\small]
     \draw  (2.75,0) coordinate (r) (-1.5,0) coordinate (l) ;
    \foreach \x in {-4,-3,-2,-1,1,2,3,4}
    {\draw (20*\x+10:2) node[bnode] (\x) {};};
     \foreach \x in {0}
    {\draw (10+20*\x:2) node[rnode] (0) {};};
     \foreach \x in {-4,-3,...,3}
    {\pgfmathtruncatemacro{\y}{\x + 1} \draw (\x) to (\y) ;};
    \draw (55:1.6) node {$\ell_w$}; \draw (-40:1.6) node {$\ell_u$};
    \draw (10:1.9) node[left] {$\ell_v$=1};
    \draw (-10:2.25) node {$u$}; \draw (10:2.25) node {$v$}; \draw (30:2.25) node {$w$};
  \end{tikzpicture}}\hspace*{5mm}
\subfigure[$u$ and $v$ both singleton]{\hspace*{5mm}\begin{tikzpicture}[font=\small]
     \draw  (2.75,0) coordinate (r) (-1.5,0) coordinate (l) ;
    \foreach \x in {-4,-3,-2,0}
    {\draw (20*\x+10:2) node[rnode] (\x) {};};
     \foreach \x in {-1,1,2,3,4}
    {\draw (10+20*\x:2) node[bnode] (\x) {};};
     \foreach \x in {-4,...,3}
    {\pgfmathtruncatemacro{\y}{\x + 1} \draw (\x) to (\y) ;};
    \draw (55:1.6) node {$\ell_w$}; \draw (-50:1.6) node {$\ell_q$};
    \draw (30:2.25) node {$w$};
    \draw (-10:1.9) node[left] {$\ell_u$=1}; \draw (10:1.9) node[left] {$\ell_v$=1};
    \draw (-10:2.25) node {$u$}; \draw (10:2.25) node {$v$};\draw (-30:2.25) node {$q$};
  \end{tikzpicture}}}
  \caption{Cases for discordant edge $u\/v$}\label{cycle:fig02}
\end{figure}
Note that $k$ can change only if $uv\in B\cup C$. Now let
\[ \lambda_{uv}\ =\ \left\{ \begin{array}{ll}
                       \sqrt{\ell_u}+\sqrt{\ell_v}, & uv\in A\,; \\
                       \sqrt{\ell_u}+\frac12\sqrt{\ell_v}, & uv\in B\,; \\
                       \frac12\sqrt{\ell_u}+\frac12\sqrt{\ell_v}, & uv\in C\,.
                     \end{array}\right.\qquad
\delta_{uv}\ =\ \left\{ \begin{array}{ll}
                       \delta_u+\delta_v, & uv\in A\,; \\
                       \delta_u+\frac12\delta_v, & uv\in B\,; \\
                       \frac12\delta_u+\frac12\delta_v, & uv\in C\,.
                     \end{array}\right.\]
Each singleton is in two discordant edges, all other discordant vertices in one, and each run is bounded by two discordant vertices. Therefore
\[\psi\,=\,\tfrac12\sum_{v\in D}\sqrt{\ell_v}\,=\,
\sum_{uv\in K}\lambda_{uv}\,,\qquad\delta\,=\,\frac1\kappa\sum_{v\in D}\delta_{v}\,=\,\frac1\kappa\sum_{uv\in K}\delta_{uv}\,.\]
We will show that $\delta_{uv}<0$ for all $uv\in K$. We consider cases (A), (B) and (C) separately.
So far, the analysis is identical for pull and push voting. Now we must distinguish them. First we consider the push process.

\subsubsection*{Push voting}
\begin{enumerate}[label=(\Alph*)]
  \item
  \begin{eqnarray*}
  \delta_v &=&\sqrt{\ell_v+1}-\sqrt{\ell_v}+\sqrt{\ell_u-1}-\sqrt{\ell_u}, \\ \delta_u &= &\sqrt{\ell_v-1}-\sqrt{\ell_v}+\sqrt{\ell_u+1}-\sqrt{\ell_u}.
  \end{eqnarray*}
  Hence $\delta_{uv}=(\sqrt{\ell_v+1}+\sqrt{\ell_v-1}-2\sqrt{\ell_v})+(\sqrt{\ell_u+1}+\sqrt{\ell_u-1}-2\sqrt{\ell_u})
      \,\leq\, -\frac14(\ell_v^{-3/2}+\ell_u^{-3/2})$, using Lemma~\ref{cycle:lem02}.\label{discord:A}
\begin{lemma}\label{cycle:lem02}
For all $\ell\geq1$, $\sqrt{\ell+1}+\sqrt{\ell-1}\,\leq\, 2\sqrt{\ell}-\frac14\ell^{-3/2}$.
\end{lemma}
\begin{proof}\vspace{-\topsep}
First, we prove the inequality $\sqrt{1+x}+\sqrt{1-x}\,\leq\, 2-\frac14 x^2$, for all $x\leq1$.
By squaring both sides, the inequality is true if
$2+2\sqrt{1-x^2}\,\leq\, 4-x^2+\tfrac1{16}x^4$.
This is true if $\sqrt{1-y}\,\leq\, 1-\tfrac12y$, with $y=x^2$.
Squaring both sides, this is $1-y^2\,\leq\, 1-y^2+\tfrac14y^4$, which is clearly true.
Now, letting $x=1/\ell$, $\sqrt{\ell+1}+\sqrt{\ell-1}\,\leq\, 2\sqrt{\ell}-\frac14\ell^{-3/2}$ is equivalent to $\sqrt{1+x}+\sqrt{1-x}\,\leq\, 2-\frac14 x^2$ with $x\leq 1$.
\end{proof}
  \item Let $u,w$ be the discordant neighbours of $v$. Then
\[
\delta_v=\frac12(\sqrt{\ell_u-1}-\sqrt{\ell_u}+\sqrt{2}-1+\sqrt{\ell_w-1}
-\sqrt{\ell_w}+\sqrt{2}-1)\]
  Since $\sqrt{\ell-1}\leq \sqrt{\ell}$, $\delta_v\leq \sqrt{2}-1$. Also
\[
\delta_u=\sqrt{\ell_w+\ell_u+1}-\sqrt{\ell_w}-\sqrt{\ell_u}-1\leq \sqrt3-3,
\]
using Lemma~\ref{cycle:lem04}.
      Thus
\[
\delta_{uv}\leq \frac12(\sqrt{2}-1)+\sqrt3-3<-1\,\leq\, -\frac12(\ell_v^{-3/2}+\ell_u^{-3/2}).
\]
\begin{lemma}\label{cycle:lem04}
For all $\ell_1,\ell_2\geq1$, $\sqrt{\ell_1}+\sqrt{\ell_2}+1\,\geq\, \sqrt{\ell_1+\ell_2+1}+(3-\sqrt{3})$.
\end{lemma}
\begin{proof}\vspace{-\topsep}
Consider $f(\ell_1,\ell_2)=\sqrt{\ell_1}+\sqrt{\ell_2}+1-\sqrt{\ell_1+\ell_2+1}+(\sqrt3-3)$. Then, for all $\ell_1,\ell_2>0$,
\[ \frac{\partial f}{\partial \ell_i}\,=\, \frac{1}{2\sqrt{\ell_i}}-\frac{1}{2\sqrt{\ell_1+\ell_2+1}}\,>\,0
\qquad(i=1,2)\,.\]
Hence $f(\ell_1,\ell_2)\geq f(1,1)=0$ for all $\ell_1,\ell_2\geq 1$.
\end{proof}\label{discord:B}
\item
Let $u,w$ be the discordant neighbours of $v$, and $v,q$ the discordant neighbours of $u$. Then
\[
\delta_v=\frac12(\sqrt{\ell_w-1} -\sqrt{\ell_w}+\sqrt{2}
-1+\sqrt{\ell_q+2}-\sqrt{\ell_q}-2).
\]
  Now $\sqrt{\ell-1}\leq \sqrt{\ell}$ and $\sqrt{\ell+2}-\sqrt{\ell}-2\leq \sqrt3-3$, using Lemma~\ref{cycle:lem04} with $\ell_1=1$. Thus $\delta_v\leq \frac12(\sqrt{2}-1+\sqrt3-3)<-0.425$. Similarly $\delta_u<-0.425$, so $\delta_{uv}<-0.425\,\leq\, -\frac15(\ell_v^{-3/2}+\ell_u^{-3/2})$.\label{discord:C}
\end{enumerate}
Hence we have $\delta_{uv}<-\frac15(\ell_v^{-3/2}+\ell_u^{-3/2})$ for all $uv\in K$, so
\[ \delta\,=\,\frac1\kappa\sum_{v\in D}\delta_{v}\,=\,\frac1\kappa\sum_{uv\in K}\delta_{uv}\,\leq\,-\frac{1}{5\kappa}
\sum_{uv\in K}(\ell_v^{-3/2}+\ell_u^{-3/2})\,<\,-\frac{1}{5\kappa}
\sum_{v\in D}\ell_v^{-3/2}\,.\]

Thus
\[ \E[\psi(X_{t+1})]\,<\,\psi(X_{t})-\frac{1}{5\kappa}
\sum_{v\in D}\ell_v^{-3/2}\,.\]
Since $f(x)=x^{-3}$ is a convex function, $\E[f(X)] \ge f(\E[X])$ by Jensen's inequality~\cite[6.6]{Williams}, so
\[\frac{1}{\kappa}\sum_{v\in D}\ell_v^{-3/2}\,\geq\,\Big(\frac{1}{\kappa}\sum_{v\in D}\sqrt{\ell_v}\Big)^{-3}\,
=\,\Big(\frac{\kappa}{2\psi(X_t)}\Big)^{3}\,\geq\,\Big(\frac{k}{2\psi(X_t)}\Big)^{3}\,,\]
Therefore,
\begin{equation}\label{not15} \E[\psi(X_{t+1})]\,<\,\psi(X_t)-\frac{1}{5}\Big(\frac{k}{2\psi(X_t)}\Big)^{3}\,=\,
\psi(X_t)-\frac{k^3}{40\hspace{0.5pt}\psi(X_t)^{3}}\,.
\end{equation}
Hence, using Lemma~\ref{cycle:lem01},
\begin{equation}\label{cycle:eq01}
\E[\psi(X_{t+1})]-\E[\psi(X_t)]\,\leq\,-\tfrac1{40}k^3/(kn)^{3/2}
\,=\,-\tfrac1{40}(k/n)^{3/2}\,.
\end{equation}
Recall that phase $r$ of the process, during which the number of runs is $k=2r$, is the interval $[t_r,t_{r-1})$, for $r\in[r_0]$, and let $\varphi_r=\E[\psi(X_{t_r})]$. Since $r_0=\frac12k(X_0)$, $t_{r_0}=0$ and, since $r(X_T)=k(X_T)=0$, $t_{0}=T$ and $\varphi_0=0$. Let $m_r=\E[t_{r-1}-t_r]$, for $r\in[r_0]$ and $\gamma_r=\tfrac1{40}m_r(2r/n)^{3/2}$. Then~\eqref{cycle:lem01} implies that $\psi(X_t)+(t-t_{r-1})\gamma_r$ is a supermartingale~\cite[10.3]{Williams} during phase $r$, and $t_r$ is a stopping time. Then the optional stopping theorem~\cite[10.10]{Williams} implies that
\[ \varphi_{r-1}+\gamma_rm_r\,=\,\E[\psi(X_{t_{r-1}})+\gamma_r(t_r-t_{r-1})]\,\leq\,\E[\psi(X_{t_r})]\,=\,\varphi_r\,,\]
which implies
\begin{equation}\label{cycle:eq02}
 \varphi_r-\varphi_{r-1}\,\geq\,\gamma_r m_r\,=\, \tfrac1{40}m_r(2r/n)^{3/2}\qquad(r\in[r_0])\,.
\end{equation}
Note, in particular, that $\varphi_r\geq\varphi_{r-1}$ for all $r\in[r_0]$.

From Lemma~\ref{cycle:lem01}, $\varphi_r\leq\sqrt{2rn}$. Then, from \eqref{cycle:eq02}, we have $m_r\leq 40\sqrt{2rn}(2r/n)^{-3/2}=20n^2/r$. Thus
\[ \E[T]\,=\,\sum_{j=1}^{r_0}m_j\,\leq\, 20n^2\sum_{j=1}^{r_0}1/j\,<\,20n^2(\ln r_0+1)\,. \]
Since $r_0\leq n/2$, this gives an absolute bound of $20n^2\ln(e n/2)=O(n^2\log n)$. However, we can improve this with a more careful analysis.

Let $x_r=\varphi_r-\varphi_{r-1}\geq0$, for $r\in[r_0]$, so $\varphi_r=\sum_{i=1}^{r}x_j\leq\sqrt{2rn}$. Also, from~\eqref{cycle:eq02}, we have $m_r\leq 40x_r(n/2r)^{3/2}=10\sqrt2\, n^{3/2}x_r/r^{3/2}$, so $\E[T]=\sum_{j=1}^{r_0}m_j< 10\sqrt2\, n^{3/2}\sum_{j=1}^{r_0}x_r/r^{3/2}$.

Thus $E[T]$ is bounded above by $T^\star$, the optimal value of the following linear program.
\begin{equation}\label{cycle:eq03}
  \begin{array}{r@{\ }l@{\ }ll}
     \ T^\star &\,=\,&\max\ 10\sqrt2\, n^2\sum_{r=1}^{r_0}x_r/r^{3/2} &  \\[1.5ex]
    \mbox{such that\ \,}  \sum_{j=1}^{r}x_j\,&\leq&\, \sqrt{2rn}& (r\in[r_0])\\[1ex]
    \qquad x_j\,&\geq&0\,&(j\in[r_0])\,.
    \end{array}
\end{equation}
This linear program can be solved easily by a greedy procedure. In fact, it is a polymatroidal linear program~\cite{DW73}, but we will give a self-contained proof for this simple case, using linear programming duality.
\begin{lemma}
Let $0<b_1<b_2<\cdots<b_\nu$ and $c_1>c_2>\cdots>c_\nu>0$. Then the linear program $\max \sum_{j=1}^\nu c_jx_j$ subject to $\sum_{j=1}^rx_j\leq b_r$, $x_r\geq 0$ $(r\in[\nu])$ has optimal solution $x_1=b_1$, $x_j=b_j-b_{j-1}$ ($j=2,3,\ldots,\nu$).
\end{lemma}
\begin{proof}
This solution has objective function value $c_1b_1+c_2(b_2-b_1)+\cdots+c_\nu(b_\nu-b_{\nu-1})$.
The dual linear program is $\min \sum_{i=1}^\nu b_iy_i$ subject to $\sum_{i=j}^\nu y_i\geq c_j$, $y_j\geq 0$ $(j\in[\nu])$, and has feasible solution $y_\nu=c_\nu$, $y_j=c_j-c_{j+1}$  ($j\in[\nu-1]$). Then the dual objective function has value $b_\nu c_\nu+b_{\nu-1}(c_{\nu-1}-c_\nu)+\cdots+b_1(c_1-c_2)$. However,
\[ c_1b_1+c_2(b_2-b_1)+\cdots+c_\nu(b_\nu-b_{\nu-1})\,=\,b_\nu c_\nu+b_{\nu-1}(c_{\nu-1}-c_\nu)+\cdots+b_1(c_1-c_2)\,. \]
Since the objective function values are equal, it follows that the two solutions are optimal in the primal and dual respectively.
\end{proof}

Thus, the optimal solution to~\eqref{cycle:eq03} is $x_r=\sqrt{2r}-\sqrt{2(r-1)}=\sqrt{2r}(1-\sqrt{1-1/r})\leq \sqrt{2/r}$, for $r\in[r_0]$, since $1-y\leq\sqrt{1-y}$ for $0\leq y\leq 1$. Thus
\[ T^\star\,\leq\, 10\sqrt2\, n^2\sum_{j=1}^{r_0}x_r/r^{3/2}\,\leq\,10\sqrt2\, n^2\sum_{j=1}^{r_0}\sqrt2/\big(\sqrt{r}\, r^{3/2}\big)\,=\, 20\onept n^2\sum_{r=1}^{r_0}1/r^{2}\,<\,(10\pi^2/3)n^2\,,\]
since $\sum_{r=1}^{\infty}1/r^{2}=\pi^2/6$. Thus we have an absolute bound of $\E[T]=O(n^2)$.

\subsubsection*{Pull voting}
The case of pull voting is similar, but the calculations for cases \ref{discord:A}--\ref{discord:C} are changed as follows.
\begin{enumerate}[label=(\Alph*${}'$)]
  \item The analysis for this case is identical to \ref{discord:A}, except that $\delta_u$ and $\delta_v$ are interchanged. Hence $\delta_{uv}
      \,\leq\, -\frac14(\ell_v^{-3/2}+\ell_u^{-3/2})$, as before.\label{discord:A'}
  \item $\delta_v=\sqrt{\ell_u+\ell_w+1}-\sqrt{\ell_u}-\sqrt{\ell_w}-1\leq \sqrt3-3$, using Lemma~\ref{cycle:lem04}.  Also $\delta_u=\sqrt2+\sqrt{\ell_u-1}-\sqrt{\ell_u}-1\leq \sqrt2-1$.
      Thus $\delta_{uv}\leq \sqrt{2}-1+\frac12(\sqrt3-3)<-0.22\,\leq\, -\frac1{10}(\ell_v^{-3/2}+\ell_u^{-3/2})$.\label{discord:B'}
  \item $\delta_v=\sqrt{\ell_w+2}-\sqrt{\ell_w}-2<\sqrt3-3$, from Lemma~\ref{cycle:lem04} with $\ell=1$. Similarly $\delta_u<\sqrt3-3$, so $\delta_{uv}\leq\sqrt3-3\,<\,-1.25\,<\,-\frac12 (\ell_v^{-3/2}+\ell_u^{-3/2})$.\label{discord:C'}
\end{enumerate}
Hence we have $\delta_{uv}<-\frac1{10}(\ell_v^{-3/2}+\ell_u^{-3/2})$ for all $uv\in K$, whereas we had $\delta_{uv}<-\frac1{5}(\ell_v^{-3/2}+\ell_u^{-3/2})$ for push voting. Thus the estimated rate of convergence is only half that for push voting. The rest of the analysis follows the same lines as before, except that the convergence time estimates are doubled. However, we may still conclude that $\E[T]=O(n^2)$.
\subsubsection*{Lower bound}
Suppose $G$ is an $n$-cycle, with $n=2\nu$ even, and the push or pull process starts with $X_0(i)=0$ $(i=1,\ldots,\nu)$, $X_0(i)=1$ $(i=\nu+1,\ldots,n)$. Thus $k=2$ and $\ell_1=\ell_2=\nu$. See Fig.~\ref{cycle:fig03}. At each step before convergence, there are two discordant edges, four discordant vertices, and the push and pull processes proceed identically.
\begin{figure}[H]
\centerline{%
    \begin{tikzpicture}[rotate=70,font=\scriptsize]
    \foreach \x in {1,...,9}
    {\draw (20*\x+10:2cm) node[bnode] (\x) {};}
     \foreach \x in {10,...,18}
    {\draw (20*\x+10:2cm) node[rnode] (\x) {};}
    \draw (18) to (1); \foreach \x in {1,...,17}
    {\pgfmathtruncatemacro{\y}{\x + 1} \draw (\x) to (\y) ;};
  \end{tikzpicture}}
  \caption{Lower bound configuration}\label{cycle:fig03}
\end{figure}
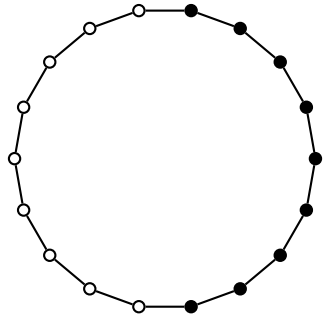
Let $L_t$ be the length of (say) the red run at step $t$, so $L_0=\nu$, $L_T\in\{0,n\}$.  At each step before convergence,  we have $k(X_t)=2$, $L_{t+1}\gets L_t-1$ with probability $\nicefrac12$, and $L_{t+1}\gets L_t+1$ with probability $\nicefrac12$. Thus $L_t$ is a symmetric simple random walk. The number of runs $k(X_t)$ can only be reduced from two to zero if either $L_t=1$ or $L_t=n-1$, when one of the runs is a singleton. Thus $\E[T]$ is bounded below by the expected time for a symmetric simple random walk started at $\nu$ to reach either $1$ or $(n-1)$. This is well known~\cite[XIV.3]{Feller}, and is exactly $(\nu-1)^2=\Omega(n^2)$. Therefore the expected convergence time for either the push or pull process is $\Theta(n^2)$.


\section{Voting on the star graph $S_n$}

Let $(r,b,X)$ denote the coloring of the star graph $S_n$ on $n$ vertices
in which there are $r$ red vertices $b=n-r$ blue vertices. The central vertex has colour $X \in\{R,B\}$.

\subsection*{Push voting on the star}

In the case of the push process, the transitions from state $(r,b,R)$ are
to state $(r+1,b-1,R)$ with probability $1/(b+1)$ and to state $(r-1,b+1,B)$
with probability $b/(b+1)$. The transitions from state $(r-1,b+1,B)$ are to $(r,b,R)$ with probability $(r-1)/r$ and to $(r-2,b+2,B)$ with probability $1/r$.
For the purposes of discussion we group the states $(r,R)=(r,b,R)$ and $(r-1,B)=(r-1,b+1,B)$ into a single pseudo-state $S(r)$. The  transitions probabilities  within or between $S(r+1)$ or $S(r-1)$ are shown in Figure \ref{Star-push}, and are derived as follows:

\begin{figure}
\begin{center}
\begin{tikzpicture}[state/.style={rounded rectangle,minimum size=5mm,
minimum width=3cm,very thick,draw=black!50},xscale=0.8,yscale=0.9]
\draw (-6,1) node (U1) [state] {$r+1,b-1,R$};\draw (-6,-1) node (L1) [state] {$r,\,b,\,B$};
\draw (0,1) node (U0) [state] {$r,\,b,\,R$};\draw (0,-1) node (L0) [state] {$r-1,b+1,B$};
\draw (6,1) node (U-1) [state] {$r-1,b+1,R$};\draw (6,-1) node (L-1) [state] {$r-2,b+2,B$};
\draw[darkgray,very thick] (-9,1)edge[<-](U1) (U1)edge[<-](U0) (U0)edge[<-](U-1) (9,1)edge[->](U-1);
\draw[darkgray,very thick]  (-9,-1)edge[->](L1) (L0)edge[<-](L1) (L-1)edge[<-](L0)  (9,-1)edge[<-](L-1);
\draw (-8,-2) rectangle (-4,2) (-2,-2) rectangle (2,2) (4,-2) rectangle (8,2);
\draw (-6,2.4) node{\large$S(r+1)$} (0,2.4) node{\large$S(r)$} (6,2.4) node{\large$S(r-1)$} ;
\draw[very thick] (U1)edge[darkgray,->,bend right=30](L1) (L1)edge[darkgray,->,bend right=30](U1) ;
\draw[very thick] (U0)edge[darkgray,<-,bend right=30](L0) (0.85,0) node{$\frac{b}{b+1}$}
(L0)edge[darkgray,<-,bend right=30](U0) (-0.85,0) node{$\frac{r-1}{r}$};
\draw[very thick] (U-1)edge[darkgray,->,bend right=30](L-1) (L-1)edge[darkgray,->,bend right=30](U-1);
\draw (-3,1.4) node {$\frac{1}{b+1}$} (-3,-0.6) node {$\frac1{r+1}$};
\draw (3,1.4) node {$\frac1{b+2}$} (3,-0.6) node {$\frac1r$};
\end{tikzpicture}
\end{center}
\caption{Pseudo-states for the push process}\label{Star-push}
\end{figure}
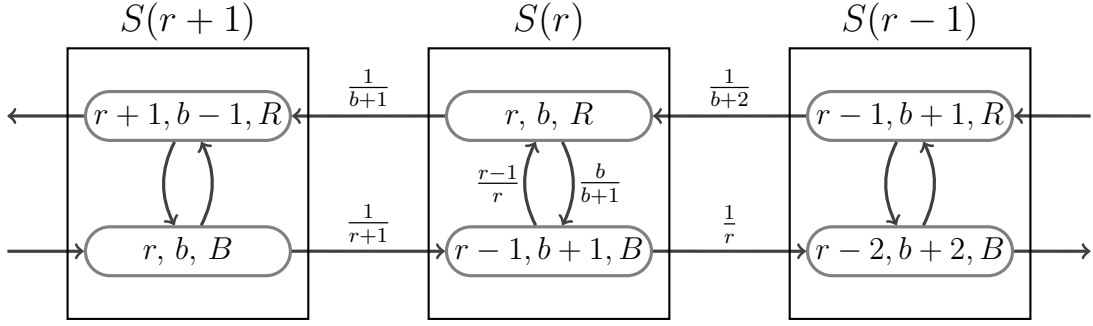

 Let $X, Y \in \{ R,B\}$. For a particle occupying a state (of colour) $X$ in $S(r)$ let $P_X(Y,r)$ be the probability of exit from $S(r)$ via state $Y$.
 For example $P_R(R,r)$ is the probability that a particle starting  at $(r,R)$ eventually exits from $S(r)$ via state $(r,R)$ to state $(r+1,R)$ in $S(r+1)$. Thus
\[
P_R(R,r)=  \frac{1}{b+1}
\brac{ 1+ \frac{b}{b+1}\frac{r-1}{r}+\cdots+\brac{\frac{b}{b+1}\frac{r-1}{r}}^k+\cdots},
\]
so that
\[
P_R(R,r)=\frac{1}{b+1} \frac{1}{1-[b(r-1)/(b+1)r]}=\frac{r}{n}.
\]
Similarly let
 $P_B(R,r)$ be the probability that a particle currently at $(r-1,B)$ in $S(r)$ moves from  $S(r)$ to $(r+1,R)$ in $S(r+1)$. Then
\[
P_B(R,r)= \frac{r-1}{r}P_R(R,r)=\frac{r-1}{n}.
\]
In summary, starting from state $X \in \{R,B\}$ of $S(r)$, for $1 \le r \le  n-1$ the transition probability $p_X(r)$ from $S(r)$ to $S(r+1)$ (resp.  transition probability $p_X(b)$ from $S(r)$ to $S(r-1)$) is given by
\begin{equation}\label{eqn:p_idef2}
p_X(r) =
\frac{r-1_{(X=B)}}{n}, \qquad p_X(b) =
\frac{b+1_{(X=B)}}{n}.
\end{equation}
States $(0,B)$ (i.e. $S(0)$) and  $(n,R)$ (i.e. $S(n)$) are absorbing.

Let $i=\max(r,b)-n/2$. To obtain lower and upper bounds
on  the number of transitions {\em between pseudo-states $S(r)$ before absorption}, we can couple the process with a biassed random walk on the
line $L=\{0,1,...,n/2\}$ with a reflecting barrier at $0$ and an absorbing barrier at $n/2$. We assume $n$ is even here. For $0 < i < n/2$, let $p_i$ be the  probability
 of a transition from  $i$ to $i+1$ on $L$, and let $q_i=1-p_i$ be the  probability
 of a transition from $i$ to $i-1$.
It follows from \eqref{eqn:p_idef2} that to obtain  bounds on  the number of transitions  between pseudo-states $S(r)$ before absorption we can use a value of  $p_i$ given by
\beq{pidef}
p_i=1/2+(i+1)/n \quad \text{Lower bound}, \qquad p_i=1/2 +(i-1)/n \quad \text{Upper bound}.
\eeq

We next consider the number of loops,  for example $(r,R)\to (r-1,B)  \to (r,R)$,
  made within $S(r)$ before exit. For a particle starting from state $X$ of $S(r)$ let $C_{XY}=C_{XY}(r)$ be the number of loops before exit at state $Y$. Let $\l=\frac{b}{b+1} \frac{r-1}{r}$  and $\r=\l/(1-\l)^2$, then
\[
\E C_{RR}= \sum_{k \ge 0} \frac{1}{b+1} k \l^k= \frac{1}{b+1}\frac{\l}{(1-\l)^2}=
\r \frac{1}{b+1}.
\]
Similarly,
\[
\E C_{BR}=\r\frac{r-1}{r(b+1)}, \qquad \E C_{RB}= \r \frac{b}{r(b+1)}, \qquad \E C_{BB}= \r \frac{1}{r}.
\]
 The conditional expectations $\mu_{XY}(r)=\E C_{XY}(r)/P_X(Y,r)$ are given by
\beq{C}
\mu_{XY}(r)=
\begin{cases}
\r \frac{n}{r}\frac{1}{b+1},& XY=RR\\
\r \frac{n}{r}\frac{1}{b+1},& XY=BR\\
\r  \frac{n}{n-r} \frac{b}{r(b+1)},& XY=RB\\
\r \frac{n}{n-r+1} \frac{1}{r},& XY=BB
\end{cases}
.
\eeq
The value of $\r=(rb(r-1)(b+1))/n^2$.
In particular if $b,r= \ooi n/2$ then, whatever colours $X,Y$
\beq{muval}
\mu_{XY}(r)= \ooi \frac{n}{4}.
\eeq

Let $N=n/2$. Starting from $r=b=n/2$
let $T'_N$ be the  number of transitions between states $S(r)$ to reach $\max(r,b)=N+n/2$. Referring to \eqref{pidef}, we consider a biassed random walk with transition probabilities of $Z = \max\{r,b\}-n/2$  given by
\begin{equation}\label{pdef}
p_i =
\begin{cases}
1, & \text{if }i = 0 \\
1/2+i/n+\d/n, & \text{if }i \in \{1, \ldots, n/2-1\} \\
0 , & \text{if }i = n/2
\end{cases}
,
\end{equation}
where we set $\d=1$ for a lower bound  on the number of steps $T'$ to absorption,
and $\d=-1$ for an upper bound.

The walk in \eqref{pdef} is the push chain $Z_t$  with transitions given
by \eqref{eqn:p_idef} as analysed Section \ref{BDC}.
Referring to \eqref{eqn:p_idef} and \eqref{ETM} we set  $\d=0$ for a lower bound on
$\E_0T_M$. For  $M=N^{3/4}$, from Lemma \ref{qwik},
\[
\E_0T_M \ge \Th(1) \sum_{i= \sqrt N}^M \frac Ni \ge \Th(N) \log \frac{M}{\sqrt N}
=\Th (n \log n).
\]
For all states $i=\sqrt N,...,N^{3/4}$, the corresponding value of $r=\ooi n/2$.
Referring to \eqref{muval}, whatever the type of transition $XY$ between $S(r)$ and neighbouring states, $\mu_{XY}(r) =\ooi n/4$. Let $\mu = \min_{X,Y}(\mu_{XY}(r): n/2 \le r \le M)$, then $\mu \ge n/5$.
As $\E_0T_N \ge \E_0T_M=\Th(n \log n)$ we have that
\[
\E T(\text{\em Push})\ge \mu\; \E_0 T_M  =\Om (n^2 \log n).
\]
The upper bound follows by a similar argument. Put $\d=-1$ in \eqref{eqn:p_idef},
and use Lemma \ref{slow}.

\subsection*{Pull voting on the star}

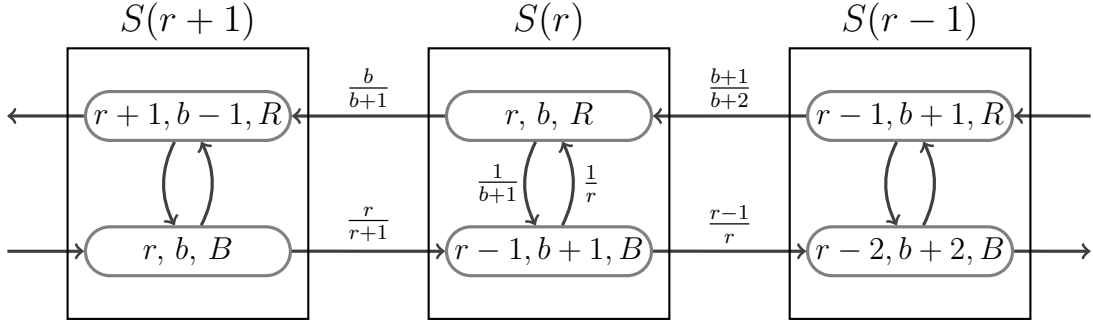
\begin{figure}
\begin{center}
\begin{tikzpicture}[state/.style={rounded rectangle,minimum size=5mm,
minimum width=3cm,very thick,draw=black!50},xscale=0.8,yscale=0.9]
\draw (-6,1) node (U1) [state] {$r+1,b-1,R$};\draw (-6,-1) node (L1) [state] {$r,\,b,\,B$};
\draw (0,1) node (U0) [state] {$r,\,b,\,R$};\draw (0,-1) node (L0) [state] {$r-1,b+1,B$};
\draw (6,1) node (U-1) [state] {$r-1,b+1,R$};\draw (6,-1) node (L-1) [state] {$r-2,b+2,B$};
\draw[darkgray,very thick]  (-9,1)edge[<-](U1) (U1)edge[<-](U0) (U0)edge[<-](U-1) (9,1)edge[->](U-1);
\draw[darkgray,very thick] (-9,-1)edge[->](L1) (L0)edge[<-](L1) (L-1)edge[<-](L0)  (9,-1)edge[<-](L-1);
\draw (-8,-2) rectangle (-4,2) (-2,-2) rectangle (2,2) (4,-2) rectangle (8,2);
\draw (-6,2.4) node{\large$S(r+1)$} (0,2.4) node{\large$S(r)$} (6,2.4) node{\large$S(r-1)$} ;
\draw[very thick] (U1)edge[darkgray,->,bend right=30](L1) (L1)edge[darkgray,->,bend right=30](U1) ;
\draw[very thick] (U0)edge[darkgray,->,bend right=30](L0) (0.7,0) node{$\frac{1}{r}$}
(L0)edge[darkgray,->,bend right=30](U0) (-0.85,0) node{$\frac{1}{b+1}$};
\draw[very thick] (U-1)edge[darkgray,->,bend right=30](L-1) (L-1)edge[darkgray,->,bend right=30](U-1);
\draw (-3,1.4) node {$\frac{b}{b+1}$} (-3,-0.6) node {$\frac{r}{r+1}$};
\draw (3,1.4) node {$\frac{b+1}{b+2}$} (3,-0.6) node {$\frac{r-1}{r}$};
\end{tikzpicture}
\end{center}
\caption{Pseudo-states for the pull process}\label{Star-pull}
\end{figure}
 As before, we group the states $(r,R)=(r,b,R)$ and $(r-1,B)=(r-1,b+1,B)$ into a single pseudo-state $S(r)$. The  transitions probabilities  within or between $S(r+1)$ or $S(r-1)$ are shown in Figure \ref{Star-pull}, and are obtained by calculations similar to the push case.
  In the final pseudo-state $S(n)$ on the left, the state $(n,0,R)$ is absorbing, and so the state $(n-1,1,B)$ cannot be reached. As an initial state, $(n-1,1,B)$ goes to $(n-2,2,B)$ with probability 1.

The pull process seems much easier to analyse.
Suppose the star currently has a red central vertex, and we are in  state $(r,b,R)$ of  $S(r)$.
The probability of
a direct transition from $(r,b,R)$ to $(r+1,b-1,R)$ is $b/(b+1)$.
This occurs when a blue leaf vertex is chosen and pulls the colour of the red
central vertex. We say  a {\em run} is  a sequence of  transitions which leave the colour of the central vertex unchanged.
Let $\r(r,x, R)$ be run given by the sequence of  transitions
\[
(r,b,R)\to (r+1,b-1,R) \to \cdots \to (x-1,n-x+1,R) \to (x,n-x,R).
\]
Then
\[
\Pr(\r(r,x,R))= \frac{n-r}{n-r+1} \frac{n-r-1}{n-r} \cdots \frac{n-x+1}{n-x+2}
= \frac{n-x+1}{n-r+1}.
\]
The probability a run starting at $(r,n-r,R)$ run finishes by absorption at $(n,0,R)$ is
\[
\Pr(\r(r,n,R))=\frac{1}{n-r+1} \ge \frac 1n.
\]
Each run is terminated by absorption, or by a change of colour of the central vertex, say from $R$ to $B$. In the latter case, this marks the start of a new run (possibly of length zero) in the opposite direction.
Starting from $(r,n-r,R)$, let $X$ be the  number of changes of colour of the central vertex from $R$ to $B$, or vice versa,
 before  absorbtion at  $(n,0,R)$ or $(0,n,B)$. Let $Y$ be the winning step for a sequence of independent trials with success probability $p=1/n$.
Then $\E X\le \E Y =n$. Each run  has a length between zero and  $n$, so $\E T(\text{\em Pull})=O(n^2)$.

\section{Voting on the double star}

\subsection*{Push voting on the double star}
A \emph{double star} $S_{2n+2}^\star$ comprises two stars $S_1,S_2$, each with $n$ leaves, and their central vertices $c_1,c_2$ joined by an edge.  Let $X_t:V\to\{R,B\}$ identify the colours of the vertices $v\in V$ at time $t$. See Fig.~\ref{double:fig01}. We will show that the convergence time for the push process on $S_{2n+2}^\star$ can be exponential in $n$.
\begin{theorem}
The push process on the double star with $2n+2$ vertices has worst case convergence time $\Omega(2^{2n/5})$.
\end{theorem}
\begin{proof}
We will assume that the initial configuration for the process has $X_0(v)=B$ $(v\in S_1)$, and  $X_0(v)=R$ $(v\in S_2)$.
Then, for convergence to occur, we must have either $X(v)=R$ $(\forall v\in S_1)$, or $X(v)=B$ $(\forall v\in S_2)$. Without loss of generality, we suppose $S_1$ that must be recoloured $R$, and temporarily restrict attention to $S_1$.

Let $r_t=|\{v\in S_1\setminus c_1:X_t(v)=R\}$ be the number of leaves in $S_1$ which are coloured $R$, and hence $(n-r_t)$ leaves are coloured $B$. We make no assumption about $X_t(c_1)$ or $X_t(c_2)$. See~Fig.~\ref{double:fig02}.

\begin{figure}[H]
\centerline{%
    {\begin{tikzpicture}
    \foreach \x in {4,...,6}
    {\draw (20*\x-90:2.5cm) node[bnode] (\x1) {};}
    \foreach \x in {7,...,14}
    {\draw (20*\x-90:2.5cm) node[rnode] (\x1) {};}
    \draw (0,0) node[circle,draw,inner sep=0pt,fill=lightgray,minimum size=2.5mm] (c1) {} ;
    \draw (-0.35,-0.35) node {\large$c_1$};
    \foreach \x in {4,...,14}
    {\draw (c1)--(\x1) ;}
    \draw (10:3cm) node (r){{\large$r$}} ;
    \draw[line width=1,->] (r)edge(28:3cm) (r)edge(-8:3cm) ;
    \draw (c1)--(0,-1.5) ;
  \end{tikzpicture}}}
  \caption{$S_1$ with $r$ leaves coloured $R$}\label{double:fig02}
\end{figure}
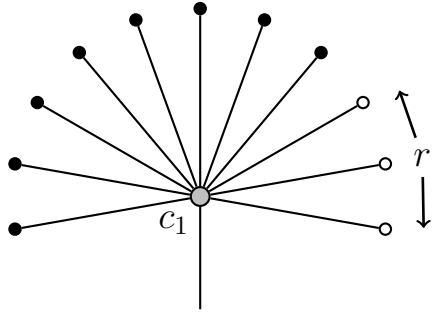

Now, if $r_{t-1}=r$, at step $t$  either $r_{t}\gets r+1$, $r_{t}\gets r-1$, $c_1$ changes colour, or the step involves $S_2$. We discard all steps which involve $S_2$ or $c_1$, and consider the time $t$ as changing only when either $r_{t+1}\gets r_t+1$ or $r_{t+1}\gets r_t-1$. Thus $t$ is a lower bound on the duration of the process.

We will upper bound $\Pr(r_{t+1}= r+1)$, when $r_t=r$. This event occurs only when $c_1$ is chosen, and will be maximised when $X_t(c_1)=R$, since otherwise $c_1$ must first change colour. It is also maximised when $X_t(c_2)=R$, since then $c_1c_2$ cannot be chosen as a discordant edge. However, $c_1$ may be recoloured $B,R$ any number of times, $k$ say, between $t$ and $t+1$. The probability that $c_1$ is recoloured $B$ is at most $(n-r+1)/(n-r+2)$, when $c_2$ is coloured $B$.
Subsequent to this,  the probability that $c_1$ is recoloured $R$ is at most $(r+1)/(r+2)$, when $c_2$ is coloured $R$.
\begin{align*}
\Pr(r_{t+1}= r+1 \mid r_t=r)\ &\leq\ \frac{1}{n-r+1}\sum_{k=0}^\infty\left(\frac{r+1}{r+2}\,\frac{n-r+1}{n-r+2}\right)^k\\
&=\ \frac{1}{n-r+1}\,\left(1-\dfrac{r+1}{r+2}\,\dfrac{n-r+1}{n-r+2}\right)^{-1}\\
&=\ \frac{(r+2)(n-r+2)}{(n+3)(n-r+1)}\\
&\leq\ \frac{r+3}{n+3},\mbox{\ \ if }r\leq (n-1)/2.
\end{align*}
Since the only alternative is that $r_{t+1}=r-1$, when $r\leq (n-1)/2$, we also have
\[ \Pr(r_{t+1}= r-1 \mid r_t=r)\ =\ 1-\Pr(r_{t+1}= r+1)\ \geq\ 1-\frac{r+3}{n+3}\ =\ \frac{n-r}{n+3}.\]
Now $\Pr(r_{t+1}= r+1 \mid r_t=r)\leq (r+3)/(n+3)\leq \nicefrac15$ if $r\leq (n-12)/5$. Let $\nu=\lfloor(n-12)/5\rfloor$. Thus, in the range $0\leq r_t\leq \nu$, the process $r_t$ is dominated by a random walk $Z_t$ with $\Pr(Z_{t+1}=r+1 \mid Z_t=r)=\nicefrac15$, $\Pr(Z_{t+1}=r-1\mid Z_t=r)=\nicefrac45$.
Let a \emph{trial} of this process be the sequence of $T$ steps, starting with $Z_0=1$, until either of the events $E_0:Z_T=0$ or $E_\nu:Z_T=\nu$ occurs. From~\cite[p.314]{Feller}, we have
\[ \Pr(E_\nu)\ = \ \frac{3}{4^\nu-1}\ \leq\ 4^{1-\nu}\quad \mbox{for }\nu>1. \]
Let $E^{1,k}_\nu$ be the event that $E_\nu$ ever occurs in $k$ trials.
Thus $\Pr(E^{1,k}_\nu)\leq k4^{1-\nu}=4k/4^\nu$. The corresponding event $E^{2,k}_\nu$ in $S_2$ is that $n-r_t=\nu$ occurs in $k$ trials, and so similarly $\Pr(E^{2,k}_\nu)\leq k4^{1-\nu}$. Let $E^k_\nu=E^{1,k}_\nu\vee E^{2,k}_\nu$, so
\[ \Pr(E^{k}_\nu)\,=\,\Pr(E^{1,k}_\nu\vee E^{2,k}_\nu)\,\leq\,\Pr(E^{1,k}_\nu) + \Pr(E^{2,k}_\nu)\,\leq\, 8k/4^{\nu}\,=\,k/2^{2\nu-3}.\]
Clearly convergence requires $E^{k}_\nu$ to have occurred. However, if $k\leq 2^{2(\nu-5)}$, $E^{k}_\nu$ occurs with probability at most $\nicefrac14$.  Thus we need at least $\Omega(4^\nu)=\Omega(2^{2n/5})$ trials before there is any appreciable probability of convergence. Hence $\Omega(2^{2n/5})$ is a lower bound on the time for convergence with high probability.
\end{proof}
For a double star $S_N^*$ on $N=2n+2$ vertices, it follows that for the push process $\E T =\Om(2^{N/5})$, as stated in Theorem \ref{Sn}.

\subsection*{Pull voting on the double star}
\begin{lemma}
Let $T$ be the expected time to complete discordant pull voting on the double star
of $2n+2$ vertices. Then for any starting configuration $\E T=O(n^4)$.
\end{lemma}

\begin{proof}
Our proof mimics that for pull voting on the star graph.
If the centers $c_1,c_2$ are the same colour (say red)
we call the central edge monochromatic.
If the central vertices are both red (e.g.),  a run is
a sequence of steps in which a blue leaf vertex is chosen at each step and pulls the red colour from one of the central vertices.

Let $r_1,b_1$ be the red and blue leaves
in $S_1$ (resp. $r_2, b_2$ in $S_2$). Let $b_1+b_2=b$.
Let $\r(b,k \mid R)$ be the probability of a run of length at least $k\ge 0$ given the central vertices are red. The probability  that a  central vertex is recoloured at the next step is $\r(b,0 \mid R)=2/(b+2)$. The required probabilities are
\[
\r(b,k \mid R)= \left\{
\begin{array}{ll}
\frac{b}{b+2}&k=1\\
&\\
\frac{b(b-1)}{(b+2)(b+1)}&k=2\\
&\\
\frac{(b-k+2)(b-k+1)}{(b+2)(b+1)}& k=3,...,b-1\\
&\\
\frac{2}{(b+1)(b+2)} & k=b
\end{array}
\right.
\]
Before cancelation of terms, for $k \ge 3$ the  expression  for $\r(b,k \mid R)$ is
\[
\frac{b}{b+2}\frac{b-1}{b+1} \frac{b-2}{b} \cdots \qquad \cdots\frac{b-(k-3)}{b-(k-3)+2}
\frac{b-(k-2)}{b-(k-2)+2}\frac{b-(k-1)}{b-(k-1)+2}.
\]
The cases $k=1,2$ are  given by the first two terms of this expression.

If the central edge monochromatic.
then the probability $P$ to
finish voting without recoloring either of $c_1, c_2$ is  $P=\r(b,b\mid R) \ge 1/n^2$. Let $\mu'$ be an upper bound on the expected number of runs required for an exit (i.e. for the entire
colouring to be monochromatic). Then $\mu'\le  1/P=n^2$.

If the central edge is not monochromatic,  e.g. $c_1$ is red and $c_2$ is blue,
let the probability of becoming monochromatic in a given step be $\f(r_1,b_1,r_2,b_2)$.
Thus
\[
\f(r_1,b_1,r_2,b_2) \ge \min\left\{ \frac{2}{b_1+r_2+2}, \frac{2}{r_1+b_2+2}\right\}
  \ge \frac{2}{2n+2}=\frac{1}{n+1}.
\]
Let $\mu$ be an upper bound on the expected wait for the central edge to become monochromatic. Then $\mu \le n+1$.

The number of steps in any  run is
at most $s=2n+1$. Thus for the pull process
\[
\E T \le \mu \mu' s= (n+1)\;n^2\;(2n+1) = O(n^4).
\]
\end{proof}


\section{Voting on the barbell graph}
The barbell or dumbbell graph of $n$ vertices, $B_{2n}$, is given by two disjoint cliques $S_1$ and $S_2$  of size $n$ joined by a single edge $e$. It has $N=2n$ vertices and $2 {n\choose 2} +1$ edges.

\subsection*{Push voting on the barbell}
We start with the following configuration: all vertices in $S_1$ are red, and all vertices in $S_2$ are blue.
Let $T$ the first time when the whole of $S_1$ is blue (or $S_2$ is red). Clearly $T$ is less than (or equal to) the time to reach consensus. For simplicity, we just look at $S_1$ and assume the final colour of $S_1$ (and $S_2$) is blue. Suppose that $N_t$ is the number of blue vertices in $S_1$, where initially $N_t = 0$.
Let $M_t$ be the number of discordant vertices, where $M_0=2$. When $1 \le N_t \leq n/5-9$ then $M_t \ge n$, and
\begin{eqnarray*}
\Pr(N_{t+1} = N_t+1| N_t) \leq (N_t+1)/M_t \leq (N_t+1)/n \le 1/5,\\
\Pr(N_{t+1} = N_t-1|N_t) = (n-N_t)/M_t \ge 2/5.
\end{eqnarray*}
In the regime $1\leq N_t \leq n/5-9$,  $N_{t}$ is dominated by a process $N'_t$ with
\begin{eqnarray}
\Pr(N'_{t+1} = N'_t+1| N'_t) = 1/5, \nonumber\\
\Pr(N'_{t+1} = N'_t-1|N'_t) = 2/5, \nonumber\\
\Pr(N'_{t+1} = N'_t|N'_t)= 2/5. \label{loopy}
\end{eqnarray}
Let $Z$ be $N'_t$ observed when $N'$ changes, and thus we ignore the loop steps given
by \eqref{loopy}. In which case, the probability $p$ that $Z$ increases by one is $p=1/3$, and the probability $q$ that $Z$ decreases by one is $q=2/3$.
We now follow the
analysis for push voting on the double star.
Let a \emph{trial} of this process be the sequence of $T$ steps, starting with $Z_0=1$, until either of the events $E_0:Z_T=0$ or $E_\nu:Z_T=\nu$ occurs. From~\cite[p.314]{Feller}, we have
\[
\Pr(E_\nu)\ = \ \frac{1}{2^\nu-1}\ \leq\ 2^{1-\nu}\quad \mbox{for }\nu>1.
\]
From now on, the same argument used for the double star works here. We just repeat the conclusion that $\E T = \Omega(2^\nu)=\Omega(2^{n/5})=\Omega(2^{N/10})$, where $N=2n$ is the total number of vertices.

\ignore{
Assuming $Z_0 = N_0 = 1$, trial is is a succession of $T$ steps such that $Z_T = \nu = \lfloor n/5-9 \rfloor$ or $Z_T = 0$. When the trial finishes, it returns to $Z_0 = 1$, just for convenience. Let $E$ the event that $Z_T = \nu$, then from Feller we have that $\Pr(E) \leq 4^{1-\nu}$. Let $E_k$ be the event that after $k$ trial at least one of them finish with $Z_T = \nu$, then $\Pr(E_k) \leq k4^{1-\nu}$. Since we are running two processes at the same time (The process in $S_1$ and $S_2$) we denote by $E_k^{1,2}$ the event that at least one of the two process satisfied the event $E_k$, then $\Pr(E_k^{1,2}) \leq 2k4^{1-\nu}$. If $k\leq 2^{2(\nu-5)}$ then $\Pr(E_k^{1,2}) \leq 1/4$, thus with probability at most $1/4$, hence with probability at least $3/4$ we need at least $\Theta(4^\nu)$ trials in both processes to at least one of them to get $Z_T = \nu$, then the expected number of trial to finish at least one of them is $\Theta(4^\nu)$, giving the desired result.
}

\subsection*{Pull voting on the barbell}

We suppose we have reached a configuration in which all vertices except one are red. Suppose the unique blue vertex  is in $S_1$. We modify our process so that the system reaches consensus faster. To do that, in each round we only select vertices in $S_1$, and assume the final colour will be red. If the final colour would be blue, then we must also recolor all of $S_2$.
Even if the vertex $c_1$ of the bridge edge $e=(c_1, c_2)$ is blue, the interaction between $S_1$ and $S_2$ does not affect the outcome. If $S_1$ is not in consensus then each vertex in $S_1$ has at least one discordant neighbour in $S_1$, so the (red)  opinions in $S_2$ will not affect the outcome.

We use a result from the proof of Theorem \ref{Kn} for $K_n$ as given in Section \ref{KN}. Inequality \eqref{Knfrom1} shows that the expected time for pull voting to reach consensus in $K_n$, when all but one vertex is red is $\Omega(2^n)$. So, the time to finish in our modified process is $\Omega(2^n)=\Omega(2^{N/2})$.

\ignore{
\section{Discordant voting: Simulation results}\label{Sim}
For the cases of discordant voting on the cycle and star graphs,
simulations were made for the push, pull and oblivious protocols and for a range of vertex set sizes ($n$). The initial vertex coloring was random.
The averages are based on a minimum of 10 simulations  for each case ($n$,  protocol, graph). The maximum graph size studied was $n=300$ (star), and $n=350$ (cylcle).
 We included the oblivious protocol  in our experiments as a control case.
 For the oblivious protocol, and a random initial colouring, the average time to consensus should be about $\E T \sim n^2/4$ for both graphs.

 The figures speak for themselves (see Figure \ref{star-sim}). Experimentally, for a cycle of size $n$, there is little difference between the protocols in the average time taken to reach consensus. However, for a star graph, the average time to reach consensus is longer for the push protocol, even for $n=150$.

\begin{figure}[H]
\label{star-sim}
\begin{center}
\includegraphics[width=220pt]{CyclefigureV2.pdf}
\includegraphics[width=220pt]{Starfigure.pdf}
\end{center}
\caption{Simulation results.
 Push (blue), Pull (red), Oblivious (green)}
\end{figure}

\newpage
}

\end{document}